%% file: Main.tex
\documentclass[11pt]{article}
\usepackage{etoolbox}
\usepackage{amssymb}
\usepackage{amsmath}
\usepackage{amsthm}
\usepackage{forloop}
\usepackage[paper=letterpaper,margin=1in]{geometry}
\usepackage[ruled,vlined,linesnumbered]{algorithm2e}
\usepackage{paralist}
\usepackage{nicefrac}
\usepackage{algpseudocode}
\usepackage{multirow}
\usepackage{array}
\usepackage{thmtools}
\usepackage{hyperref}
\usepackage{caption}
\usepackage{subcaption}
\usepackage{tkz-graph}

\hypersetup{
  colorlinks   = true, 
  urlcolor     = blue, 
  linkcolor    = blue, 
  citecolor   = red 
}

\newtheorem{theorem}{Theorem}[section]
\newtheorem{lemma}[theorem]{Lemma}

\newtheorem{corollary}[theorem]{Corollary}

\newtheorem{observation}[theorem]{Observation}

\newcommand{\defcal}[1]{\expandafter\newcommand\csname c#1\endcsname{{\mathcal{#1}}}}
\newcommand{\defbb}[1]{\expandafter\newcommand\csname b#1\endcsname{{\mathbb{#1}}}}
\newcounter{calBbCounter}
\forLoop{1}{26}{calBbCounter}{
    \edef\letter{\Alph{calBbCounter}}
    \expandafter\defcal\letter
		\expandafter\defbb\letter
}

    \makeatletter
\def\@fnsymbol#1{\ensuremath{\ifcase#1\or *\or **\or \dagger\or \ddagger\or
   \mathsection\or \mathparagraph\or \|\or \dagger\dagger
   \or \ddagger\ddagger \else\@ctrerr\fi}}

\newcommand{\eps}{\varepsilon}
\newcommand{\ie}{{\it i.e.}}
\newcommand{\eg}{{\it e.g.}}
\newcommand{\trianglefreealg}{{\textnormal{\textsc{TriangleFreeAlg}}}}
\newcommand{\email}[1]{{\href{mailto:#1}{#1}}}

\newcommand{\naive}{{na\"{i}ve}}
\newcommand{\scomp}{{\text{(\#single)}}}
\newcommand{\dcomp}{{\text{(\#double)}}}
\newcommand{\tcomp}{{\text{(\#triangle)}}}
\newcommand{\dnpotential}{{\text{(\#component-free)}}}
\newcommand{\ddpotential}{{\text{(\#component-component)}}}
\newcommand{\dspotential}{{\text{(\#single-component)}}}
\newcommand{\sspotential}{{\text{(\#single-single)}}}
\newcommand{\dmpotential}{{\text{(\#middle)}}}
\newcommand{\ftriangles}{{\text{(\#non-$M^*$-triangles)}}}

\newcommand{\misseddd}{{\text{(\#lost-component-component)}}}
\newcommand{\ftrianglesdn}{{\text{(\#non-$M^*$-triangles-$A_1$)}}}
\newcommand{\ftrianglesdd}{{\text{(\#non-$M^*$-triangles-$A_2$)}}}
\newcommand{\MMF}{{\texttt{MMF}}}

\author{Moran Feldman\thanks{Department of Computer Science, University of Haifa, E-mail: \email{moranfe@cs.haifa.ac.il}.} \and Ariel Szarf\thanks{Department of Mathematics and Computer Science, Open University of Israel, E-mail: \email{aszarf@gmail.com}.}}

\title{Maximum Matching sans Maximal Matching: A New Approach for Finding Maximum Matchings in the Data Stream Model}

\begin{document}

\maketitle
\pagenumbering{Alph}
\thispagestyle{empty}
\input{Abstract}
\newpage

\pagenumbering{arabic}

\input{Introduction}
\input{RelatedWork}
\input{Preliminaries}
\input{TwoPassComponents}
\input{TriangleFreeImprovementPaths}
\input{GeneralImprovementPaths}

\appendix

\input{ThreePassComponents}

\bibliographystyle{plain}
\bibliography{Matching}

\end{document}

%% file: Abstract.tex

\begin{abstract}
The problem of finding a maximum size matching in a graph (known as the \emph{maximum matching} problem) is one of the most classical problems in computer science. Despite a significant body of work dedicated to the study of this problem in the data stream model, the state-of-the-art single-pass semi-streaming algorithm for it is still a simple greedy algorithm that computes a maximal matching, and this way obtains $\nicefrac{1}{2}$-approximation. Some previous works described two/three-pass algorithms that improve over this approximation ratio by using their second and third passes to improve the above mentioned maximal matching. One contribution of this paper continuous this line of work by presenting new three-pass semi-streaming algorithms that work along these lines and obtain improved approximation ratios of $0.6111$ and $0.5694$ for triangle-free and general graphs, respectively.

Unfortunately, a recent work~\cite{DBLP:journals/corr/abs-2107-07841} shows that the strategy of constructing a maximal matching in the first pass and then improving it in further passes has limitations. Additionally, this technique is unlikely to get us closer to single-pass semi-streaming algorithms obtaining a better than $\nicefrac{1}{2}$-approximation. Therefore, it is interesting to come up with algorithms that do something else with their first pass (we term such algorithms non-maximal-matching-first algorithms). No such algorithms are currently known (to the best of our knowledge), and the main contribution of this paper is describing such algorithms that obtain approximation ratios of $0.5384$ and $0.5555$ in two and three passes, respectively, for general graphs (the result for three passes improves over the previous state-of-the-art, but is worse than the result of this paper mentioned in the previous paragraph for general graphs). The improvements obtained by these results are, unfortunately, numerically not very impressive, but the main importance (in our opinion) of these results is in demonstrating the potential of non-maximal-matching-first algorithms.

\end{abstract}

%% file: Introduction.tex
\section{Introduction} \label{sec:introduction}

The problem of finding a maximum size matching in a graph (known as the \emph{maximum matching} problem) is one of the most classical problems in computer science, and many polynomial time algorithms have been designed for it over the years (see, \eg, \cite{BalinskiG91,edmonds1965maximum,HopcroftK73}). Due to its central role, the maximum matching problem is often one of the first problems considered when new computational models are suggested. One such model is the data stream model, which is motivated by Big-Data applications, and has been the subject of an enormous amount of research over the last couple of decades.

In the data stream model, the algorithm receives the input in the form of a stream which it can read sequentially, but due to memory restrictions, the algorithm can store only a small part of this stream. This means that the algorithm has to process (in some sense) the input stream while reading it, and never gets an opportunity to see all the parts of the input at the same time. Traditional algorithms for this model, known as \emph{streaming algorithms}, are allowed only memory that is poly-logarithmic in the natural parameters of the problem. Obtaining a streaming algorithm for a problem is very desirable, but is often not possible. In particular, many graph problems provably do not admit streaming algorithms, and the maximum matching problem is among these problems if one would like an algorithm for the problem to output an (approximately) maximum matching because such a matching might be of linear size in the number of vertices. Nevertheless, non-trivial streaming algorithms have been designed for the maximum matching problem when only the (approximate) size of a maximum matching is desired (see Section~\ref{ssc:related_work} for details).

The resistance of many graph problems to streaming algorithms has motivated Feigenbaum et al.~\cite{FeigenbaumKMSZ04} to suggest semi-streaming algorithms, which are algorithms for the data stream model that are allowed a space complexity of $O(n \log^c n)$ for some constant $c \geq 0$, where $n$ is the number of vertices in the graph. Such algorithms turn out to be a sweet-spot that on the one hand allows many results of interest, and on the other hand, does not lead to triviality because $O(n \log^c n)$ is less than the space necessary for storing the input graph (unless this graph is very sparse). In particular, Feigenbaum et al.~\cite{FeigenbaumKMSZ04} observed that one can obtain $\nicefrac{1}{2}$-approximation for the maximum matching problem using a simple semi-streaming algorithm that greedily constructs a maximal matching.\footnote{A maximal matching is a matching that is inclusion-wise maximal, and it is well-known that the size of any maximal matching is a $\nicefrac{1}{2}$-approximation for the size of a maximum matching.}

The above $\nicefrac{1}{2}$-approximation semi-streaming algorithm for the maximum matching problem also has the desirable property that it reads the input stream only once (\ie, it makes a single pass over it). Surprisingly, no single-pass semi-streaming algorithm improving over the approximation ratio of this simple algorithm was suggested in the decade and a half that has already passed since the work of~\cite{FeigenbaumKMSZ04} (in contrast, Kapralov~\cite{DBLP:conf/soda/Kapralov21} showed that no such algorithm can have an approximation ratio better than $1/(1+\ln 2) \approx 0.59$, improving over previous inapproximability results due to~\cite{DBLP:conf/soda/GoelKK12,DBLP:conf/soda/Kapralov13}). Given this lack of progress, interest arose in obtaining improved approximation ratios for relaxed versions of the above problem. Perhaps, the simplest such relaxation is to allow the algorithm to make multiple (sequential) passes over the input stream. Some works tried to understand the approximation ratio that can be obtained as the number of passes grows (but remains constant)---see Section~\ref{ssc:related_work} for more detail. Another line of work is interested in studying semi-streaming algorithms with very few passes (usually two or three).

\begin{table}
\caption{The state-of-the-art approximation ratios for semi-streaming algorithms using two or three passes, and our improvements over these ratios (the number to the right of each improvement is the number of the theorem formally stating it).} \label{tbl:results}

\begin{center}
\begin{tabular}{l|l|cr|cr}
\multicolumn{1}{c|}{\textbf{Number}}&\multicolumn{1}{c|}{\textbf{Type of}}&\multicolumn{2}{c|}{\multirow{2}{*}{\textbf{State-of-the-Art}}}&\multicolumn{2}{c}{\multirow{2}{*}{\textbf{This Paper}}}\\
\multicolumn{1}{c|}{\textbf{of Passes}}&\multicolumn{1}{c|}{\textbf{Graphs}}&&&&\\
\hline
\rule{0pt}{2.6ex}\multirow{3}{*}{Two-Pass}&Bipartite&$2 - \sqrt{2} \approx \frac{1}{2} +\frac{1}{11.66} \approx 0.5857$&\cite{Konrad18}&-&\\[1mm]
&Triangle-Free&$\frac{1}{2} + \frac{1}{16} = 0.5625$&\cite{KaleT17}&-&\\[1mm]
&General&$\frac{1}{2} + \frac{1}{32} = 0.53125$&\cite{KaleT17}&$\frac{1}{2} + \frac{1}{26} \approx 0.5385$&(\ref{thm:two_pass})\\[1mm]
\hline
\rule{0pt}{2.6ex}\multirow{3}{*}{Three-Pass}&Bipartite&$0.6067 \approx \frac{1}{2} + \frac{1}{9.37}$ & \cite{Konrad18}&$\frac{1}{2} + \frac{1}{9} \approx 0.6111$&(\ref{thm:improvement_paths_triangle_free})\\[1mm]
&Triangle-Free&$\frac{1}{2} + \frac{1}{10} = 0.6$&\cite{KaleT17}&$\frac{1}{2} + \frac{1}{9} \approx 0.6111$&(\ref{thm:improvement_paths_triangle_free})\\[1mm]
&General&$\frac{1}{2} + \frac{81}{1600} \approx \frac{1}{2} + \frac{1}{19.753} \approx 0.5506$&\cite{KaleT17}&$\frac{1}{2} + \frac{1}{14.4} \approx 0.5694$&(\ref{thm:improvement_paths_general})
\end{tabular}
\end{center}
\end{table}


The state-of-the-art results for the last line of work are summarized in Table~\ref{tbl:results}. We note that beside the state-of-the-art results for general input graphs, Table~\ref{tbl:results} also gives improved results for bipartite and triangle-free graphs. All the known results in this line of work (to the best of our knowledge) start by greedily constructing a maximal matching during the first pass over the input stream, and then augmenting this matching in the subsequent passes. Recently, Konrad and Naidu~\cite{DBLP:journals/corr/abs-2107-07841} showed that this technique has limitations (specifically, even for bipartite graphs, a two-pass semi-streaming algorithm based on this technique cannot obtain a better than $2/3$-approximation). Additionally, and arguably more importantly, multi-pass algorithms that use their first pass for constructing a maximal matching are unlikely to be a step towards a single-pass semi-streaming algorithm with a better than $\nicefrac{1}{2}$-approximation guarantee.

Given the above observations, it is natural to believe that the future of the study of semi-streaming algorithms for the maximum matching problem lies in algorithms that use their first pass in a more sophisticated way than simply constructing the traditional maximal matching. We term such algorithms \emph{non-maximal-matching-first algorithms} (or non-{\MMF} algorithms for short). In this paper, we present the first non-{\MMF} algorithms, which leads to improvements over the state-of-the-art both for two and three passes. Admittedly, the improvements we obtain are numerically not very impressive, but their main importance (in our opinion) is in demonstrating the potential of non-{\MMF} algorithms.

To intuitively understand our non-{\MMF} algorithms, one should note that greedily constructing a maximal matching is equivalent to greedily constructing a graph whose connected components are of size at most $2$ (where the size of a connected component is defined as the number of vertices in it). Therefore, a natural generalization is to greedily construct in the first pass a graph whose connected components are of size at most $3$. There are two intuitive advantages for doing that compared to constructing a maximal matching.
\begin{itemize}
	\item If many connected components end up to be of size $2$ rather than $3$, then it is not possible for many of the edges of a maximum matching to intersect only a single connected component of the constructed graph; and therefore, the constructed graph must have many connected components compared to the size of a maximum matching.
	\item A connected component of size $3$ can contribute two edges to the output matching if it is ``augmented'' during in the next passes with a single additional edge. In contrast, doing the same with a connected component of size $2$ requires ``augmenting'' it with two additional edges. It is important to note that there is a significant conceptual difference between an augmentation of a connected component with one or two edges. Augmenting a connected component with two edges requires finding pairs of edges that augment the \emph{same} connected component, while augmenting with a single edge does not require such a synchronization.
\end{itemize}
Using the above ideas, we prove in Section~\ref{sec:two_passes} and Appendix~\ref{app:three_pass} the following two theorems, respectively.

\begin{restatable}{theorem}{ThmTwoPass} \label{thm:two_pass}
There exists a non-{\MMF} $2$-pass ($\nicefrac{7}{13} = \nicefrac{1}{2} + \nicefrac{1}{26}$)-approximation semi-streaming algorithm for finding a maximum size matching in a general graph.
\end{restatable}
\begin{restatable}{theorem}{ThmThreePass} \label{thm:three_pass}
There exists a non-{\MMF} $3$-pass ($\nicefrac{5}{9} = \nicefrac{1}{2} + \nicefrac{1}{18}$)-approximation semi-streaming algorithm for finding a maximum size matching in a general graph.
\end{restatable}

As mentioned above, both Theorems~\ref{thm:two_pass} and~\ref{thm:three_pass} represent an improvement over the state-of-the-art. However, it turns out that we can further improve over Theorem~\ref{thm:three_pass} using new {\MMF} algorithms (\ie, algorithms that construct a maximal matching in their first pass). This leads to the following theorems whose proofs appear in Sections~\ref{sec:triangle_free} and~\ref{sec:general}, respectively.
\newbool{skipfootnote}
\begin{restatable}{theorem}{ThmImprovementPathsTriangleFree} \label{thm:improvement_paths_triangle_free}
There exists a $3$-pass ($\nicefrac{11}{18} = \nicefrac{1}{2} + \nicefrac{1}{9}$)-approximation semi-streaming algorithm for finding a maximum size matching in a triangle-free graph.\ifbool{skipfootnote}{}{\global\booltrue{skipfootnote}\footnote{We recall that every bipartite graph is triangle-free, and therefore, the same result is obtained also for bipartite graphs.}}
\end{restatable}
\begin{restatable}{theorem}{ThmImprovementPathsGeneral} \label{thm:improvement_paths_general}
There exists a $3$-pass ($\nicefrac{1}{2} + \nicefrac{1}{14.4}$)-approximation semi-streaming algorithm for finding a maximum size matching in a general graph.
\end{restatable}

The algorithms used to prove Theorems~\ref{thm:improvement_paths_triangle_free} and~\ref{thm:improvement_paths_general} are strongly based on the algorithms suggested by Kale and Tirodkar~\cite{KaleT17}. For example, the first two passes of the algorithm suggested by Theorem~\ref{thm:improvement_paths_triangle_free} are identical to a two-pass algorithm presented by~\cite{KaleT17}, and the third pass of this algorithm is very similar to the third pass of the three-pass algorithm of~\cite{KaleT17}. Our novelty, however, is in our ability to analyze the algorithm obtained by putting these two components together.

%% file: RelatedWork.tex
\subsection{Related Work} \label{ssc:related_work}


As mentioned in Section~\ref{sec:introduction}, streaming algorithms are not appropriate for the maximum matching problem when the algorithm is required to output an (approximately) maximum matching. However, some non-trivial streaming algorithms are known for this problem when the algorithm is only required to estimate the size of the maximum matching. Kapralov et al.~\cite{DBLP:conf/soda/KapralovKS14} designed a poly-log approximation streaming algorithm for this problem under the assumption that the edges in the input stream are ordered in a uniformly random order. A different line of work~\cite{DBLP:conf/esa/CormodeJMM17,DBLP:journals/talg/EsfandiariHLMO18,DBLP:conf/approx/McGregorV16} considered graphs of bounded arboricity $\alpha$, comulating with the work of McGregor and Vorotnikova~\cite{0001V18}, who designed $(\alpha + 2)(1 + \varepsilon)$-approximation streaming algorithm for this problem requiring only $O(\varepsilon^{-2}\log n)$ space.


Recall that, to date, the best single-pass semi-streaming algorithm for the maximum matching problem is still the natural greedy algorithm, which guarantees $\nicefrac{1}{2}$-approximation. Chitnis et al.~\cite{ChitnisCEHMMV16} presented an exact single-pass algorithm for this problem. However, this algorithm requires $\tilde{O}(k^2)$ memory, where $k$ is an upper bound on the size of the maximum matching (which the algorithm needs to know upfront), and thus, this algorithm is a semi-streaming algorithm only when $k = \tilde{O}(\sqrt{n})$. Given the difficultly to improve over the guarantee of the greedy algorithm using single-pass semi-streaming algorithms, people started to considered relaxed versions of the maximum matching problem. One standard relaxation is to allow the algorithm to make multiple passes over the input stream. Section~\ref{sec:introduction} surveys algorithms of this kind that use two or three passes. Another line of work considers algorithms that assume a constant (but possibly large) number of passes. The first result of this kind was presented by Feigenbaum et al.~\cite{FeigenbaumKMSZ04} (in the same paper that also introduced the notion of semi-streaming algorithms), and guaranteed $(2/3 - \eps)$-approximation using $O(\eps^{-1} \log \eps^{-1})$ passes for bipartite graphs. Later~\cite{DBLP:conf/approx/McGregor05} showed how to obtain $(1 - \eps)$-approximation for general graphs using $(\eps^{-1})^{O(\eps^{-1})}$ passes, and the number of passes necessary to obtain this guarantee was improved by many further works (see, \eg,~\cite{DBLP:journals/topc/AhnG18,DBLP:conf/sosa/AssadiLT21,DBLP:journals/corr/abs-2106-04179}). Another standard relaxation for the maximum matching problem is to assume that the edges of the input stream appear in a uniformly random order. The state-of-the-art for this relaxation is a $(2/3+\eps_0)$-approximation single-pass semi-streaming algorithm, where $\eps_0 > 0$ is some absolute constant~\cite{DBLP:conf/icalp/AssadiB21} (see also the references therein for previous works on this relaxation).



The related maximum weight matching problem was also studied heavily in the context of the data stream model. Here, it is not immediately clear that one can obtain a constant approximation ratio using a single-pass semi-streaming algorithm. However, Feigenbaum et al.~\cite{FeigenbaumKMSZ04} presented the first such algorithm guaranteeing $\nicefrac{1}{6}$-approximation, and this ratio was improved in series of works~\cite{CrouchS14,DBLP:journals/siamdm/EpsteinLMS11,DBLP:conf/approx/McGregor05,DBLP:journals/algorithmica/Zelke12}. The current state-of-the-art for the problem is $(\nicefrac{1}{2} - \eps)$-approximation due to Paz and Schwartzman~\cite{PazS17}. Since this approximation ratio is essentially identical to the state-of-the-art for the (unweighted) maximum matching problem, any further progress on the maximum weight matching problem will imply an improvement over the guarantee of the greedy algorithm for the (unweighted) maximum matching problem.



%% file: Preliminaries.tex
\section{Preliminaries}

In this section we present the problem that we study more formally, and also introduce the notation used throughout the rest of the paper. We are interested in semi-streaming algorithms for the problem of finding a maximum size matching in a graph $G = (V, E)$ of $n$ vertices. A semi-streaming algorithm for this problem is an algorithm with a space complexity of $O(n \log^c n)$ (for some constant $c \geq 0$) that initially has no knowledge about the edges of $E$. Instead, the edges of $E$ appear sequentially in an ``input stream'', and the algorithm may make one or more passes over this input stream. In each pass the algorithm sees the edges one by one, and may do arbitrary calculations after viewing each edge. It is important to note that the space complexity allowed for the algorithm does not suffice for storing all the edges of the graph (unless the graph is very sparse), and this is the reason that the algorithm might benefit from doing multiple passes over the input stream. It is standard to assume that the vertices of $V$ are known upfront, and that each vertex of $V$ can be stored using $O(\log n)$ bits (which implies that every edge of $E$ can also be stored using this asymptotic number of bits).

Throughout the paper, we consider only unweighted graphs and matchings. We also denote by $M^*$ an arbitrary maximum matching of $G$ (\ie, an arbitrary optimal solution for our problem). Notation-wise, we treat $M^*$ (and any other matching considered in the paper) as a set of the edges included in it. Similarly, when considering a connected component $C$ of a graph, we treat it as a set of the vertices in it, which in particular, implies that $|C|$ is the number of such vertices. 

Given a set of edges $S$ or a path $P$ in a graph, we denote by $V(S)$ and $V(P)$ the set of vertices intersecting any edge of $S$ or $P$, respectively. Similarly, the set of edges included in the path $P$ is denoted by $E(P)$. Often we need to consider collections of paths (or triangles) in a given graph. For clarity, such collections are always denoted using calligraphic letters, and we extend the above notation to such collections. In other words, if $\cP$ is a collection of paths, then $V(\cP)$ and $E(\cP)$ is the set of vertices and edges, respectively, that are included in these paths. Finally, given a set $S$ of edges and a vertex $v$, we use $\deg_S(v)$ to denote the degree of the vertex $v$ in the subgraph $(V, S)$.

%% file: TwoPassComponents.tex
\section{Two-Pass Non-{\MMF} Algorithm} \label{sec:two_passes}

In this section we prove Theorem~\ref{thm:two_pass}, which we repeat below for convenience.
\ThmTwoPass*
The algorithm whose existence is guaranteed by Theorem~\ref{thm:two_pass} appears as Algorithm~\ref{alg:triangles}. In its first pass, this algorithm greedily grows a set $P$ of edges that form either triangles or partial triangles (\ie, isolated edges or paths of length $2$). For simplicity, we refer below to the connected components of $(V, P)$ that are not isolated vertices as partial triangles although, technically, they can also be full triangles. In the second pass of Algorithm~\ref{alg:triangles}, the algorithm tries to convert the partial triangles of $P$ into more involved structures in one of two ways. To understand these ways, we need to define some terms. First, we designate some of the vertices of every partial triangle as ``connection vertices''. Specifically, all the vertices of a triangle are considered connection vertices; in a path of length $2$ only the two end points are considered to be connection vertices; and finally, in an isolated edge there are no connection vertices. We refer to a partial triangle that was not converted yet into a more involved structure as a ``\naive'' partial triangle. The first way in which Algorithm~\ref{alg:triangles} tries to convert the partial triangles of $P$ into more involved structures is by greedily adding edges that connect a connection vertex of a {\naive} partial triangle with an isolated vertex. The set $A_1$ in the algorithm includes the edges that were added in this way. In parallel, the algorithm also tries a second way to convert the partial triangles of $P$ into more involved structures, which is to greedily add edges that connect a connection vertex of a {\naive} partial triangle either to a connection vertex of another {\naive} partial triangle or to an isolated vertex. The set $A_2$ in the algorithm includes the edges that were added in this way. Upon termination, Algorithm~\ref{alg:triangles} outputs a maximum matching in the set of all the edges that it kept. We recall that given a connected component $C$ of a graph, the notation $|C|$ represents the number of vertices in $C$.

\begin{algorithm}[th]
\caption{\textsc{Maximum Matching via Greedy Triangles - Two Passes}} \label{alg:triangles}
\tcp{First Pass}
Let $P \gets \varnothing$.\\
\For{every edge $e$ that arrives}
{
	\If{every connected component of the graph $(V, P \cup \{e\})$ is either a path of length at most $2$ or a triangle (cycle of size $3$)}
	{
		Add $e$ to $P$.
	}
}

\BlankLine

\tcp{Second Pass}
Let $A_1 \gets \varnothing$ and $A_2 \gets \varnothing$.\\
\For{every edge $(u, v) \not \in P$ that arrives}
{
	Let $C_u$ and $C_v$ be the connected components of $u$ and $v$, respectively, in $(V, P)$. We assume without loss of generality that $|C_u| > 1$, otherwise we swap the roles of $u$ and $v$. \tcp{Note that we cannot have $|C_u| = |C_v| = 1$ because the edge $(u, v)$ was not added to $P$ in the first pass.}
	\If{no edge of $A_1$ intersects $C_u$ and $C_v$, $|C_v| = 1$ and $u$ is a connection vertex of $C_u$}
	{
		Add the edge $(u, v)$ to $A_1$.
	}
	\If{no edge of $A_2$ intersects $C_u$ and $C_v$, $|C_v| = 1$ and $u$ is a connection vertex of $C_u$}
	{
		Add the edge $(u, v)$ to $A_2$.
	}
	\ElseIf{no edge of $A_2$ intersects $C_u$ and $C_v$, and $u$ and $v$ are connection vertices of $C_u$ and $C_v$, respectively}
	{
		Add the edge $(u, v)$ to $A_2$.
	}
}

\BlankLine

\Return{a maximum matching in the graph $(V, P \cup A_1 \cup A_2)$}.
\end{algorithm}

We begin the analysis of Algorithm~\ref{alg:triangles} by showing that it is indeed a semi-streaming algorithm.

\begin{observation}\label{obs:streaming}
Algorithm~\ref{alg:triangles} is a semi-streaming algorithm.
\end{observation}
\begin{proof}
Since every connected component of the graph $(V, P)$ is of size at most $3$, the set $P$ contains at most $n$ edges. Furthermore, each connected component of $(V, P)$ intersects at most a single edge of the set $A_1$ and at most a single edge of the set $A_2$, and therefore, each one of these sets can include at most $n/2$ edges. Hence, in total, Algorithm~\ref{alg:triangles} keeps only $O(n)$ edges.
\end{proof}

In the rest of this section we analyze the approximation ratio of Algorithm~\ref{alg:triangles}.
Recall that we use $M^*$ to denote some maximum matching of $G$. Our first objective in the analysis of the approximation ratio of Algorithm~\ref{alg:triangles} is to lower bound the number of edges of $M^*$ that can potentially be added either to $A_1$ or to $A_2$. Towards this goal, we define a charging scheme $\pi$. Under the charging scheme $\pi$, every edge $(u, v) \in M^*$ charges the connected components of $u$ and $v$ in $(V, P)$. Each one of these connected components is charged one unit by $(u, v)$, unless it is an isolated edge or an isolated vertex, in which case it is charged only half a unit or nothing by $(u, v)$, respectively. We note that when $u$ and $v$ belong to the same connected component of $(V, P)$, then this connected component is charged twice by $(u, v)$.\footnote{Intuitively, the charge assigned to the connected components of $u$ and $v$ is proportional to the ``blame'' that can be assigned to them if $(u, v)$ ends up to be outside $P$. For example, an isolated edge could not alone prevent $(u, v)$ from being added to $P$, but two such edges (one intersecting $u$ and the other intersecting $v$) could, together, prevent $(u, v)$ from being added to $P$. Therefore, we assign a charge of $\nicefrac{1}{2}$ to isolated edges. Observation~\ref{obs:lower_bound_charge} is based on this intuition.}

The following observation provides an upper bound on the total charged by all the edges of $M^*$ together. Let $\scomp$ be the number of isolated edges in $P$, $\dcomp$ be the number of connected components in $(V, P)$ that are paths of length $2$ and $\tcomp$ be the number of triangles in $P$.
\begin{observation} \label{obs:upper_bound_charge}
The total charge according to $\pi$ is at most $\scomp + 3\dcomp + 3\tcomp$.
\end{observation}
\begin{proof}
Every positive amount charged by $\pi$ is charged to some connected component of $(V, P)$ which is not an isolated vertex. Therefore, to prove the observation we only need to show that every isolated edge of $(V, P)$ is charged at most one unit, and every connected component of $(V, P)$ that is either a path of length $2$ or a triangle is charged at most $3$ units. Below we are argue that this is indeed the case.

Each connected component $C$ of $(V, P)$ can be charged at most once for every one of its vertices since the fact that $M^*$ is a matching implies that every vertex of $C$ can appear in at most a single edge of $M^*$. For isolated edges of $(V, P)$, this implies that they can be charged at most twice, and therefore, they are charged at most one unit because they are charged half a unit in each charge. Similarly, connected components of $(V, P)$ that are either paths of length $2$ or triangles contain $3$ vertices, and therefore, can be charged at most three times. Since every one of these charges is of a single unit, the total charge to each connected component of these kinds is at most $3$.
\end{proof}

To complement the last observation, let us now describe a simple lower bound on the total charging done by all the edges of $M^*$ according to $\pi$.
Let $\dnpotential$ be the number of edges of $M^*$ that connect a connection vertex of a connected component of $(V, P)$ to an isolated vertex of $(V, P)$, $\ddpotential$ be the number of edges of $M^*$ that connect connection vertices of two different connected components of $(V, P)$, $\sspotential$ be the number of edges of $M^*$ whose two end points belong to (not necessary distinct) isolated edges of $(V, P)$, $\dspotential$ be the number of edges of $M^*$ that connect a vertex of an isolated edge of $(V, P)$ with a connection vertex of some (other) connected component of $(V, P)$ and $\dmpotential$ be the number of edges that either intersect the middle vertex of a length $2$ path connected component of $(V, P)$ or are included within a triangle connected component of $(V, P)$.
\begin{observation} \label{obs:lower_bound_charge}
The total charge of all the edges of $M^*$ according to the charging scheme $\pi$ is at least $\dnpotential + 2\ddpotential + \sspotential + 1.5\dspotential + \dmpotential$.
\end{observation}
\begin{proof}
Since the edges of $M^*$ counted by $\dnpotential$ intersect a connection vertex, they must intersect a connected component of $(V, P)$ which is not an isolated vertex or an isolated edge, and therefore, they charge this connected component one unit. Hence, the total charge by all the edges counted by $\dnpotential$ is at least $\dnpotential$. Similar logic shows that the total charge by all the edges counted by $\ddpotential$, $\sspotential$, $\dspotential$ and $\dmpotential$ are at least $2\ddpotential$, $\sspotential$ , $1.5\dspotential$ and $\dmpotential$, respectively.
%
%
The observation now follows since the edges of $M^*$ counted by $\dnpotential$, $\ddpotential$, $\sspotential$, $\dmpotential$ and $\dspotential$ are distinct.
\end{proof}

Combining Observations~\ref{obs:upper_bound_charge} and~\ref{obs:lower_bound_charge}, we get the following inequality.
\begin{align} \label{eq:basic_inequality}
	\dnpotential \mspace{-50mu}&{}\mspace{50mu}+ 2\ddpotential + \sspotential \\\nonumber&{}+ 1.5\dspotential + \dmpotential
	\leq
	\scomp + 3\dcomp + 3\tcomp
	\enspace.
\end{align}
In its current form, Inequality~\eqref{eq:basic_inequality} is not very useful. We later derive from it a more convenient inequality, but before doing this we need to prove a few other inequalities. Let $\ftriangles$ denote the number of triangle connected components of $(V, P)$ that do not include any edge of $M^*$ within them.
\begin{lemma} \label{lem:auxiliary_inequalities}
The following inequalities hold
\begin{gather}
	\begin{aligned}
	\dnpotential + \ddpotential + \sspotential &\\+ \dmpotential + \dspotential &{}\geq |M^*| \label{eq:partition} \enspace,\end{aligned}\\
	\dcomp + \tcomp - \ftriangles \geq \dmpotential \label{eq:auxilary_1} \enspace,\\
	\dspotential \leq 2\sspotential + \dspotential \leq 2\scomp \label{eq:auxilary_2} \enspace,
\end{gather}
and they imply together
\begin{align*}
	\dnpotential &{}+ \ddpotential + 2\scomp \\&{}+ \dcomp + \tcomp - \ftriangles
	\geq
	|M^*|
	\enspace.
\end{align*}
\end{lemma}
\begin{proof}
Since every edge that is included in a connected component of $(V, P)$ which is a path of length $2$ must include the middle vertex of this path, every edge $e \in M^*$ that is not counted by either $\dnpotential$, $\ddpotential$, $\sspotential$, $\dspotential$ or $\dmpotential$ must either connect a vertex of an isolated edge of $(V, P)$ to an isolated vertex or connect two isolated vertices of $(V, P)$. However, such edges cannot exists. Specifically, assume towards a contradiction that $(u, v)$ is an edge of $M^*$ such that $u$ is an isolated vertex of $(V, P)$ and $v$ is either another isolated vertex of $(V, P)$ or belongs to an isolated edge of this graph. Then, the edge $(u, v)$ should have been added by Algorithm~\ref{alg:triangles} to $P$ upon arrival, which contradicts the fact that its end point $u$ ended up as an isolated vertex of $(V, P)$. Hence, every edge $e \in M^*$ is counted by either $\dnpotential$, $\ddpotential$, $\sspotential$, $\dspotential$ or $\dmpotential$, which implies Inequality~\eqref{eq:partition}.

Recall that every edge counted by $\dmpotential$ must either be included in a triangle connected component of $(V, P)$ or intersect the middle vertex of a path of length $2$ connected component of $(V, P)$. Since $M^*$ is a matching, only one edge of $M^*$ can intersect the middle vertex of a given length $2$ path or be included in a given triangle, and therefore, every edge counted by $\dmpotential$ can be associated with a distinct path of length $2$ or triangle component of $(V, P)$ that is not counted by $\ftriangles$, which implies Inequality~\eqref{eq:auxilary_1}. 

Every edge counted by $\sspotential$ touches two end-points of isolated edges of $(V, P)$. Similarly, every edge counted by $\dspotential$ intersects an end-point of an isolated edge of $(V, P)$. Since every end-point of an isolated edge of $(V, P)$ can be touched by at most a single edge of $M^*$ because $M^*$ is a matching, this implies that the number of end points of the isolated edges of $(V, P)$ is at least $2\sspotential + \dspotential$. However, this number is also equal to $2\scomp$, which implies Inequality~\eqref{eq:auxilary_2}.
\end{proof}

The last inequality in the previous lemma provides a lower bound on $\dnpotential + \ddpotential$, and one can view $\dnpotential + \ddpotential$ as a count of edges of $M^*$ that have potential to be added to $A_2$ in Algorithm~\ref{alg:triangles}. The next lemma is the promised derivative of Inequality~\eqref{eq:basic_inequality}, and it provides a lower bound on $\dnpotential$. Observe that $\dnpotential$ is a count of edges of $M^*$ that have the potential to be added to $A_1$.
\begin{lemma} \label{lem:dn_lower_bound}
$2|M^*|
	\leq
	\dnpotential - \ftriangles + 2\scomp + 4\dcomp + 4\tcomp$.
\end{lemma}
\begin{proof}
Adding twice Inequality~\eqref{eq:partition} to Inequality~\eqref{eq:basic_inequality}, we get
\begin{align*}
	&
	2|M^*| - \dnpotential - \sspotential - 0.5\dspotential - \dmpotential\\
	\leq{} &
	\scomp + 3\dcomp + 3\tcomp
	\enspace.
\end{align*}
The lemma now follows by adding Inequality~\eqref{eq:auxilary_1} and half of Inequality~\eqref{eq:auxilary_2} to the last inequality.
\end{proof}

So far we have shown lower bounds on the size of the sets of edges that have a potential to be added to $A_1$ or $A_2$ by Algorithm~\ref{alg:triangles}. Our next step is to lower bound the size of the sets $A_1$ and $A_2$ that Algorithm~\ref{alg:triangles} ends up constructing using this potential.

\begin{figure}[ht]
\begin{center}
\hfill\begin{subfigure}{0.4\textwidth}
\centering
\begin{tikzpicture}
	\SetGraphUnit{1.4}
	\tikzset{VertexStyle/.append style={minimum size=5pt, inner sep=5pt}, EdgeStyle/.style={font=\scriptsize,above,sloped,midway}}
	\Vertex[NoLabel]{A}
	\SOEA[NoLabel](A){B}
	\EA[NoLabel](B){C}
	\EA[NoLabel](C){D}
	\NOEA[NoLabel](D){E}
	\Edge[label = $M^*$](A)(B)
	\Edge[style={blue, very thick}](B)(C)
	\Edge[style={blue, very thick}](C)(D)
	\Edge[label = $M^*$](D)(E)
\end{tikzpicture}
\subcaption{Path of length 2} \label{sfg:path_of_length_2}
\end{subfigure}\hfill
\begin{subfigure}{0.4\textwidth}
\centering
\begin{tikzpicture}
	\SetGraphUnit{1.4}
	\tikzset{VertexStyle/.append style={minimum size=5pt, inner sep=5pt}, EdgeStyle/.style={font=\scriptsize,above,sloped,midway}}
	\Vertex[NoLabel]{B}
	\NOEA[NoLabel](B){A}
	\SOEA[NoLabel](A){C}
	\EA[NoLabel](C){D}
	\Edge[style={blue, very thick},label = $M^*$](A)(B)
	\Edge[style={blue, very thick}](B)(C)
	\Edge[style={blue, very thick}](A)(C)
	\Edge[label = $M^*$](C)(D)
\end{tikzpicture}
\subcaption{$M^*$-triangle} \label{sfg:M_star_triangle}
\end{subfigure}\hspace*{\fill}\\
\begin{subfigure}{0.4\textwidth}
\centering
\begin{tikzpicture}
	\SetGraphUnit{1.4}
	\tikzset{VertexStyle/.append style={minimum size=5pt, inner sep=5pt}, EdgeStyle/.style={font=\scriptsize,above,sloped,midway}}
	\Vertex[NoLabel]{B}
	\NOEA[NoLabel](B){A}
	\SOEA[NoLabel](A){C}
	\SOEA[NoLabel](C){D}
	\SOWE[NoLabel](B){E}
	\NO[NoLabel](A){F}
	\Edge[style={blue, very thick}](A)(B)
	\Edge[style={blue, very thick}](B)(C)
	\Edge[style={blue, very thick}](A)(C)
	\Edge[label = $M^*$](C)(D)
	\Edge[label = $M^*$](B)(E)
	\Edge[label = $M^*$](A)(F)
\end{tikzpicture}
\subcaption{Non-$M^*$-triangle} \label{sfg:non_M_star_triangle}
\end{subfigure}
\end{center}
\caption{A graphical study of the maximum number of $M^*$ edges counted by $\dnpotential$ that can intersect connection vertices of various types of partial triangles. Sub-figures~(\subref{sfg:path_of_length_2}) and~(\subref{sfg:M_star_triangle}) show that at most two such edges can intersect the connection vertices of a path of length $2$ and an $M^*$-triangle (\ie, a triangle that includes an edge of $M^*$). Sub-figure~(\subref{sfg:non_M_star_triangle}) shows that the connection vertices of a non-$M^*$-triangle can intersect up to $3$ edges of $M^*$.}
\end{figure}
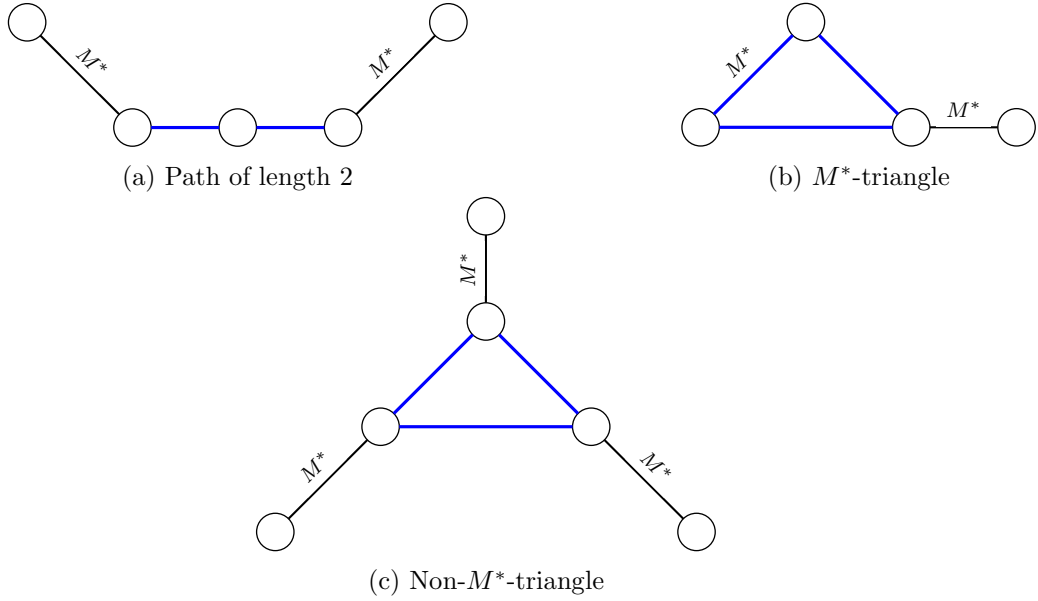

\begin{lemma} \label{lem:dn_greedy}
$3|A_1| \geq \dnpotential - \ftriangles$.
\end{lemma}
\begin{proof}
We say that an edge $e$ of $M^*$ counted by $\dnpotential$ is excluded by an edge $f \in A_1$ if $e$ and $f$ intersect the same connected component of $(V, P)$. One can observe that every edge $e$ counted by $\dnpotential$ is excluded by some edge of $A_1$ (possibly itself) when Algorithm~\ref{alg:triangles} terminates because otherwise Algorithm~\ref{alg:triangles} would have added $e$ to $A_1$, which would have resulted in $e$ excluding itself. Therefore, we can upper bound $\dnpotential$ by counting the number of edges excluded by the edges of $A_1$.

Let $(u, v)$ be an edge of $A_1$, and assume without loss of generality that $v$ is the end point of this edge which is an isolated vertex of $(V, P)$. This implies that $u$ is a connection vertex of a connected component $C_u$ of $(V, P)$ which is either a path of length $2$ or a triangle. If $C_u$ is a path of length $2$, then the edge $(u, v)$ can exclude only edges counted by $\dnpotential$ that intersect either $v$ or a connection vertex of $C_u$, and there can be only $3$ such edges because $M^*$ is a matching (see Figure~\ref{sfg:path_of_length_2}). Next, consider the case in which $C_u$ is a triangle which is not counted by $\ftriangles$. In this case there can be at most $2$ edges of $OPT$ intersecting $C_u$ (see Figure~\ref{sfg:M_star_triangle}), and therefore, even though $(u, v)$ can exclude any edge of $M^*$ intersecting $C_u$ or $v$, there can be only $3$ such edges. It remains to consider the case in which $C_u$ is a triangle counted by $\ftriangles$. In this case, $(u, v)$ can again exclude every edge of $M^*$ that intersects $C_u$ or $v$, and this time there can be at most $4$ such edges (see Figure~\ref{sfg:non_M_star_triangle}). Combining all the above, we get that the number of edges excluded by all the edges of $A_1$ is at most
\[
	3|A_1| + |\{e \in A_1 \mid \text{$e$ intersects a triangle counted by $\ftriangles$}\}|
	\enspace.
\]
As explained above, this expression is an upper bound on $\dnpotential$. Furthermore, since $A_1$ includes at most a single edge intersecting every connected component of $(V, P)$, the second term in this expression is upper bounded by $\ftriangles$. Therefore, we get
\[
	\dnpotential \leq 3|A_1| + \ftriangles
	\enspace.
\]
The lemma now follows by rearranging this inequality.
\end{proof}

The next corollary now follows by combining Lemmata~\ref{lem:dn_lower_bound} and~\ref{lem:dn_greedy}.
\begin{corollary} \label{cor:A_1_bound}
$
	2|M^*|
	\leq
	3|A_1| + 2\scomp + 4\dcomp + 4\tcomp
$.
\end{corollary}

\begin{lemma} \label{lem:dn_dd_greedy}
$4|A_2| \geq \ddpotential + \dnpotential - \ftriangles$.
\end{lemma}
\begin{proof}
The proof of Lemma~\ref{lem:dn_dd_greedy} is very similar to the proof of Lemma~\ref{lem:dn_greedy}, and therefore, we only sketch it. We first define that an edge $e \in A_2$ excludes an edge $f$ of $M^*$ counted by either $\ddpotential$ or $\dnpotential$ if they both intersect the same connected component of $(V, P)$. Like in the proof of Lemma~\ref{lem:dn_greedy}, it can be argued that $\ddpotential + \dnpotential$ is upper bounded by the total number of edges of $M^*$ excluded by the edges of $A_2$, and on the other hand, every edge $e$ of $A_2$ excludes up to $4 + T(e)$ edges, where $T(e)$ is the number of triangles counted by $\ftriangles$ that intersect $e$. Therefore,
\[
	\ddpotential + \dnpotential
	\leq
	\sum_{e \in A_2} [4 + T(e)]
	\leq
	4|A_2| + \ftriangles
	\enspace,
\]
where the second inequality holds since every connected component of $(V, P)$ intersects only a single edge of $A_2$. The lemma now follows by rearranging the last inequality.
\end{proof}

The next corollary follows by combining Lemma~\ref{lem:dn_dd_greedy} and the final inequality in Lemma~\ref{lem:auxiliary_inequalities}.
\begin{corollary} \label{cor:A_2_bound}
$
	|M^*|
	\leq
	4|A_2| + 2\scomp + \dcomp + \tcomp
$.
\end{corollary}

Let us now denote $L = \scomp + \dcomp + \tcomp + \max\{|A_1|, |A_2|\}$. We argue below that $L$ is a lower bound on the size of the solution produced by Algorithm~\ref{alg:triangles}. However, before proving this, let us show first that $L$ is large.
\begin{lemma} \label{lem:L_lower_bound}
$L \geq \nicefrac{7}{13}|M^*|$.
\end{lemma}
\begin{proof}
Plugging the definition of $L$ into Corollaries~\ref{cor:A_1_bound} and~\ref{cor:A_2_bound} yields the inequalities
\[
	2|M^*| \leq 3L - \scomp + \dcomp + \tcomp
\]
and
\[
	|M^*|	\leq 4L - 2\scomp - 3\dcomp - 3\tcomp
	\enspace.
\]
Adding the first of these inequalities three times to the second one gives
\[
	7|M^*|
	\leq
	13L - 5\scomp
	\leq
	13L
	\enspace,
\]
where the second inequality holds since $\scomp$ is non-negative by definition. The lemma now follows by rearranging the above inequality.
\end{proof}

As promised, we now argue that the size of the matching produced by Algorithm~\ref{alg:triangles} is at least $L$.
\begin{lemma} \label{lem:L_is_output}
Algorithm~\ref{alg:triangles} outputs a matching of size at least $L$.
\end{lemma}
\begin{proof}
Since Algorithm~\ref{alg:triangles} outputs a maximum matching in $(V, P \cup A_1 \cup A_2)$, to prove the lemma it suffices to show that the graph $(V, P \cup A_1)$ includes a matching of size $\scomp + \dcomp + \tcomp + |A_1|$ and the graph $(V, P \cup A_2)$ includes a matching of size $\scomp + \dcomp + \tcomp + |A_2|$. We prove below only the claim regarding $(V, P \cup A_2)$. The claim regarding $(V, P \cup A_1)$ can be proved analogously.

Let $H$ be the number of edges in $A_2$ that connect two non-isolated vertices of $(V, P)$. Then, we classify the connected components of $(V, P \cup A_2)$ as follows, and show how to build a large matching $M$ based on this classification.
\begin{itemize}
	\item $(V, P \cup A_2)$ includes $\scomp + \dcomp + \tcomp - |A_2| - H$ connected components that are (i) not an isolated node, and (ii) appear also in $(V, P)$. Each one of these connected components contains at least one edge, and therefore, can contribute some edge to $M$.
	\item $(V, P \cup A_2)$ includes $|A_2| - H$ connected components that consist of a connected component $C$ of $(V, P)$ that has connection vertices and an edge $e$ connecting a connection vertex of $C$ to an isolated vertex of $(V, P)$. One can observe that the combination of $C$ and $e$ must be either a path of length $3$ or a triangle and an edge attached to one of its vertices, and in both cases this combined connected component contains two vertex disjoint edges which it can contribute to the matching $M$.
	\item $(V, P \cup A_2)$ includes $H$ connected components that consist of two connected components $C_1, C_2$ of $(V, P)$ that have connection vertices and an edge $e$ connecting a connecting vertex of $C_1$ with a connecting vertex of $C_2$. There are three shapes that the connected component obtained in this way can take: a path of length $5$, a triangle with a path of length $3$ attached to one of its vertices or two triangles and an edge connecting them. However, one can observe that all these shapes include three vertex disjoint edges that can be contributed to the matching $M$.
\end{itemize}
By collecting from every connected component of $(V, P \cup A_2)$ the edges that it can contribute to $M$ according to the above analysis, we get a matching in $(V, P \cup A_2)$ of size at least
\begin{align*}
	&
	[\scomp + \dcomp + \tcomp - |A_2| - H] + 2[|A_2| - H] + 3H\\
	={} &
	\scomp + \dcomp + \tcomp + |A_2|
	\enspace.
	\qedhere
\end{align*}
\end{proof}

Lemmata~\ref{lem:L_lower_bound} and~\ref{lem:L_is_output} imply together the following corollary. Together with Observation~\ref{obs:streaming}, this corollary implies Theorem~\ref{thm:two_pass}.
\begin{corollary}
Algorithm~\ref{alg:triangles} is a $\nicefrac{7}{13}$-approximation algorithm.
\end{corollary}

%% file: TriangleFreeImprovementPaths.tex
\section{Three-Pass Algorithm for Triangle-Free Graphs} \label{sec:triangle_free}

In this section we prove Theorem~\ref{thm:improvement_paths_triangle_free}, which we repeat here for convenience.
\ThmImprovementPathsTriangleFree*

We refer to the algorithm whose existence is guaranteed by Theorem~\ref{thm:improvement_paths_triangle_free} as {\trianglefreealg}. In its first pass, {\trianglefreealg} constructs a maximal matching $M_0$ of $G$. Formally, the pseudocode for this pass appears as Algorithm~\ref{alg:first_pass}.

\begin{algorithm}[th]
\caption{\textsc{{\trianglefreealg} -- First Pass}} \label{alg:first_pass}
Let $M_0 \gets \varnothing$.\\
\For{every edge $e$ that arrives}
{
	Add $e$ to $M_0$ if it does not intersect any edge that already belongs to $M_0$.
}
\end{algorithm}

We say that an edge $e \in E$ is a \emph{wing} if $e$ includes exactly one vertex of $V(M_0)$. Intuitively, the reason we are interested in wings is that one can obtain an augmenting path\footnote{A path $P$ is an augmenting path for a matching $M$ if $M \oplus E(P)$ is a valid matching of size $|M| + 1$.} for $M_0$ by combining an edge $(u, v) \in M_0$ with two wings: one wing that intersects $u$ and one wing that intersects $v$. The second pass of {\trianglefreealg} grows a set $W$ of wings. Since we hope to construct multiple augmenting paths using these wings, the algorithm makes sure to limit the number of wings in $W$ that intersect any given vertex $u$ (specifically, the algorithm allows only a single wing in $W$ to intersect $u$ if $u \in V(M_0)$, and otherwise it allows up to two wings of $W$ to intersect $u$). The pseudocode of this second pass appears as Algorithm~\ref{alg:second_pass}.

Algorithm~\ref{alg:second_pass} also includes a post-processing step in which a set $\cP_1$ of augmenting paths (with respect to $M_0$) is constructed using $W$. This is done by constructing an auxiliary multi-graph $G_A$ over the vertices of $V \setminus V(M_0)$ in which there is an edge between two nodes $u, v \in V \setminus V(M_0)$ for every path $P_{u, v}$ of length $3$ in $W \cup M_0$ between them. One can note that every such path $P_{u, v}$ must be an augmenting path consisting of an edge $e \in M_0$ and two wings from $W$: one intersecting $u$ and an end-point of $e$, and the other intersecting $v$ and the other end-point of $e$. Algorithm~\ref{alg:second_pass} finds a maximum size matching $M_A$ in $G_A$, and then sets $\cP_1$ to be the collection of (augmenting) paths corresponding to the edges of $M_A$.
\begin{algorithm}[th]
\caption{\textsc{{\trianglefreealg} -- Second Pass}} \label{alg:second_pass}
Let $W \gets \varnothing$.\\
\For{every edge $e$ that arrives}
{
	\If{$e$ intersects exactly one vertex $u \in V(M_0)$}
	{
		Let $v$ denote the other end-point of $e$ (\ie, the end-point that is not $u$).\\
		\If{$\deg_W(u) < 1$ and $\deg_W(v) < 2$ \label{line:degrees_condition}}
		{
			Add $e$ to $W$.
		}
	}
}

\BlankLine

\tcp{Post-processing}
Let $G_A$ be a multi-graph over the vertices $V \setminus V(M_0)$. For every path $P_{u, v}$ of length $3$ in $W \cup M_0$ between two vertices $u, v \in V \setminus V(M_0)$, we add an edge $(u, v)$ to the graph $G_A$. \tcp{This is a multi-graph because there might be multiple such paths between a pair of vertices of $V \setminus V(M_0)$.}
Find a maximum size matching $M_A$ in $G_A$.\\
Let $\cP_1 \gets \{P_{u, v} \mid (u, v) \in M_A\}$.
\end{algorithm}

Consider now an edge $e \in M_0$ that does not appear in any path of $\cP_1$ and is connected by some wing $w \in W$ to some vertex $u \not \in V(M_0) \cup V(\cP_1)$. The pair $e, w$ can be extended into an augmenting path if one can find another wing $w'$ connecting the other end of $e$ (the end that does not intersect $w$) to a vertex $v \not \in V(M_0) \cup V(\cP_1)$ that is not $u$. The third pass of {\trianglefreealg} greedily constructs a collection $\cP_2$ of augmenting paths in this way. A pseudocode of this pass appears as Algorithm~\ref{alg:third_pass}. After completing the pass, Algorithm~\ref{alg:third_pass} returns the matching obtained by augmenting $M_0$ with the augmenting paths of $\cP_1$ and $\cP_2$.
\begin{algorithm}[th]
\DontPrintSemicolon
\caption{\textsc{{\trianglefreealg} -- Third Pass}} \label{alg:third_pass}
Let $\cP_2 \gets \varnothing$.\\
\For{every edge $w'$ that arrives}
{
	\If{there exist $4$ vertices $u, a, b, v \in V \setminus (V(\cP_1) \cup V(\cP_2))$ such that: (i) $u \not \in V(M_0)$, (ii) $w' = (u, a)$, (iii) $(a, b) \in M_0$ and (iv) $(b, v) \in W$ \label{line:P2_condition}}
	{
		Add the path $u, a, b, v$ to $\cP_2$. \tcp{Note that $u \neq v$ because otherwise $u, a, b, v$ would have been a triangle.}
	}
}

\BlankLine

\Return{$M_0 \oplus \left(\bigcup_{P \in \cP_1 \cup \cP_2} E(P)\right)$}.
\end{algorithm}

We begin the analysis of {\trianglefreealg} with the following lemma, which shows that this algorithm returns a matching, and also gives a basic lower bound on the size of this matching.
\begin{lemma} \label{lem:basic_matching_size}
The paths in $\cP_1$ and $\cP_2$ are vertex disjoint, and therefore, the output of \textnormal{\trianglefreealg} is a matching of size $|M_0| + |\cP_1| + |\cP_2|$.
\end{lemma}
\begin{proof}
Given the above discussion, it is clear that all the paths in $\cP_1 \cup \cP_2$ are augmentation paths with respect to $M_0$, which implies that the first part of the lemma indeed implies the second part. Furthermore, one can observe that the condition in Line~\ref{line:P2_condition} of Algorithm~\ref{alg:third_pass} guarantees that the paths in $\cP_2$ are vertex disjoint from each other and from the paths of $\cP_1$. Thus, to complete the proof of the lemma, it remains to argue that the paths in $\cP_1$ are also vertex disjoint.

Recall that the end-points of every path in $\cP_1$ belong to $V \setminus V(M_0)$ and the internal points of these paths belong to $V(M_0)$. Therefore, to show that the paths in $\cP_1$ are vertex disjoint, it suffices to argue this separately for their end-points and their internal nodes. Every path $P_{u, v} \in \cP_1$ corresponds to an edge $(u, v)$ in the matching $M_A$. Since the end-points of the path $P_{u, v}$ are also the end-points of this edge, we get that the paths in $\cP_1$ must have disjoint end-points because $M_A$ is a matching. Consider now some path $P_{u, v} \in \cP_1$, and let us denote the internal nodes of this path by $a$ and $b$. Since $a$ and $b$ appear only in the edge $(a, b)$ of $M_0$ (because $M_0$ is a matching), we get that if one of them belongs to a path of $\cP_1$, then the other belongs to this path as well. Furthermore, by Line~\ref{line:degrees_condition} of Algorithm~\ref{alg:second_pass}, $\deg_W(a) = \deg_W(b) = 1$, which implies that any path of $\cP_1$ that includes the nodes $a$ and $b$ as internal nodes must in fact be identical to $P_{u, v}$ itself. Hence, no two paths in $\cP_1$ share internal nodes.
\end{proof}

Using the last lemma we can also bound the space complexity of Algorithm~\ref{alg:third_pass}.

\begin{corollary} \label{cor:space_complexity}
\textnormal{\trianglefreealg} is a semi-streaming algorithm.
\end{corollary}
\begin{proof}
Aside from a constant number of other vertices and edges, {\trianglefreealg} has to store only the edges of $M_0 \cup W$ and the paths of $\cP_1 \cup \cP_2$. As these paths are of constant length (specifically, a length of $3$), to prove the corollary we only need to argue that $M_0$, $W$, $\cP_1$ and $\cP_2$ are all of size $O(n)$. Below we argue that this is indeed the case.
\begin{compactitem}
	\item $M_0$ is a matching in the graph $G$, and therefore, its size is at most $n/2$.
	\item Every edge of $W$ is a wing, and thus, has one end point in $M_0$. Since Line~\ref{line:degrees_condition} of Algorithm~\ref{alg:second_pass} guarantees that $\deg_W(u) \leq 1$ for every vertex $u \in M_0$, this implies $|W| \leq 2|M_0| \leq n$.
	\item Since the paths in $\cP_1 \cup \cP_2$ are vertex disjoint by Lemma~\ref{lem:basic_matching_size}, and each path contains $4$ vertices, the number of paths in both sets together cannot exceed $n/4$. \qedhere
\end{compactitem}
\end{proof}

It remains to analyze the approximation ratio of {\trianglefreealg}. Our analysis roughly follows the flow of the algorithm, and thus, we begin by observing that the matching $M_0$ constructed in the first pass of this algorithm is of size at least $|M^*|/2$ (recall that $M^*$ is a maximum size matching of $G$) because $M_0$ is a maximal matching of $G$ by construction.

In its second pass, {\trianglefreealg} constructs the set $W$ of wings. Our next objective is to lower bound the size of $W$. Towards this goal, we need to define $W_M$ to be the set of all edges of $M^*$ that are wings (we recall that an edge $e$ is a wing if exactly one of its end points appear in $V(M_0)$).
\begin{observation} \label{obs:W_M_size}
$|W_M| \geq 2(|M^*| - |M_0|)$.
\end{observation}
\begin{proof}
Since $M_0$ is a maximal matching, every edge of $M^*$ intersects at least one edge of $M_0$. Hence, every edge of $W_M$ includes a single end-point of an edge of $M_0$, and every edge of $M^* \setminus W_M$ includes two end-points of edges of $M_0$ (the two end-points might belong to different edges or to the same edge), which implies $|M_0| \geq  (|W_M| + 2|M^* \setminus W_M|)/2 = |M^*| - |W_M|/2$. Rearranging this inequality completes the proof of the observation.
\end{proof}


\begin{lemma} \label{lem:W_lower_bound}
$|W| \geq \tfrac{2}{3} |W_M| \geq \tfrac{4}{3} (|M^*| - |M_0|)$.
\end{lemma}
\begin{proof}
Let $I = V(M_0) \cap V(W_M)$, and let $I_F$ be the set of vertices in $I$ that do not appear in any edge of $W$. Every vertex $a \in I_F \subseteq I$ must belong to some wing $w(a) \in W_M$ by the definition $I$. However, this wing was not added to $W$ (because $a \in I_F$), which implies that the condition in Line~\ref{line:degrees_condition} of Algorithm~\ref{alg:second_pass} evaluated to FALSE when $w(a)$ arrived. Since $a$ is not covered by any edge of $W$ (\ie, $\deg_W(v) = 0$), the fact that this condition evaluated to FALSE implies that the end point of $w(a)$ that does not belong to $V(M_0)$ must have a degree of $2$ under $W$. Formally, if we denote by $u(a)$ the end point of $w(a)$ that does not belong to $V(M_0)$, then we must have $\deg_W(u(a)) = 2$.

We now observe that (i) every wing in $W_M$ contains a disjoint vertex of $V \setminus V(M_0)$ because $W_M$ is a subset of the optimal matching $M^*$, and (ii) every wing in $W$ contains only one vertex of $V \setminus V(M)$ because it is a wing. These two observations imply together
\begin{equation} \label{eq:fail_side}
	|W|
	\geq
	\sum_{a \in I_F} \deg_W(u(a))
	=
	2|I_F|
	\enspace.
\end{equation}
In contrast, since (i) every wing in $W$ contains a single vertex of $V(M_0)$, and (ii) all the vertices of $I \setminus I_F \subseteq V(M_0)$ appear in some wing of $W$,
\begin{equation} \label{eq:success_side}
	|W|
	\geq
	|I| - |I_F|
	=
	|W_M| - |I_F|
	\enspace,
\end{equation}
where the equality holds since every edge of $W_M$ is a wing, and therefore, intersects a single vertex of $V(M_0)$. The lemma now follows by adding two copies of Inequality~\eqref{eq:success_side} to Inequality~\eqref{eq:fail_side}.
\end{proof}

We now get to the analysis of the third pass of {\trianglefreealg}, and our first goal in this analysis is to identify a set of paths that have a potential (in some sense) to end up in $\cP_2$. Let $\cP'$ be the set of paths of length $3$ in $G$ that consist of a wing of $W_M$ followed by an edge of $M_0$ and then a wing of $W$. We think of the paths in $\cP'$ as directed from their $W_M$ to their $W$ edge, and consider two paths that differ only in their direction to be different paths. This is important because if there is an edge $e \in M_0$ incident to two edges $w_1, w_2 \in W \cap W_M$, then the path $w_1, e, w_2$ fulfills the requirements to belong to $\cP'$ both when $w_1$ is considered the first edge in it and when $w_2$ is considered the first edge of the path. Thus, the fact that we treat the direction of the path as part of the path's definition allows both the paths $w_1, e, w_2$ and $w_2, e, w_1$ to appear in $\cP'$.
\begin{observation} \label{obs:cPprime_size}
$|\cP'| \geq \tfrac{10}{3}|M^*| - \tfrac{16}{3}|M_0|$.
\end{observation}
\begin{proof}
Since $\deg_W(a) \leq 1$ for every vertex $a \in V(M_0)$, there are $|W|$ end-points of $M_0$ that intersect an edge of $W$. Let us denote these end-points by $V_W$, and for every end-point $a \in V_W$, we denote by $b(a)$ the other end-point of the same edge of $M_0$. Formally, $V_W = V(M_0) \cap V(W)$, and $b(a)$ is the single element of the set $\{b \mid (a, b) \in M_0\}$. One can now observe that $\cP'$ includes a (distinct) path for every wing of $W_M$ that intersect $b(a)$ for some vertex $a \in V_W$. Therefore,
\begin{align*}
	|\cP'|
	={} &
	|\{b(a) \mid a \in V_W\} \cap V(W_M)\}|\\
	\geq{} &
	|\{b(a) \mid a \in V_W\}| + |V(W_M) \cap V(M_0)\}| - |V(M_0)|
	=
	|W| + |W_M| - |V(M_0)|\\
	\geq{} &
	\tfrac{4}{3} (|M^*| - |M_0|) + 2(|M^*| - |M_0|) - |V(M_0)|
	=
	\tfrac{10}{3}|M^*| - \tfrac{16}{3}|M_0|
	\enspace,
\end{align*}
where the first equality holds since $\{b(a) \mid a \in V_W\}$ is a subset of $V(M_0)$, and the last inequality follows from Observation~\ref{obs:W_M_size} and Lemma~\ref{lem:W_lower_bound}.
\end{proof}

A path in $\cP'$ has a potential to be added to $\cP_2$ only if none of its vertices appears in $\cP_1$. Let $\cP''$ be the set of such paths (formally, $\cP'' = \{P \in \cP' \mid V(P) \cap V(\cP_1) = \varnothing\}$). The following lemma lower bounds the size of $\cP''$.
\begin{lemma} \label{lem:cPDoublePrime_size}
$|\cP''| \geq |\cP'| - 6|\cP_1| \geq \tfrac{10}{3}|M^*| - \tfrac{16}{3}|M_0| - 6|\cP_1|$.
\end{lemma}
\begin{proof}
The second inequality of the lemma follows from Observation~\ref{obs:cPprime_size}, and therefore, we concentrate on proving the first inequality. Towards this goal, assume that $P' \in \cP'$ is a path that intersects with a path $P_1 \in \cP_1$ on an internal vertex. Since the middle edge of both paths is an edge of $M_0$, this implies that the two paths intersect on both their internal vertices. Furthermore, since both end-edges of $P_1$ and one end-edge of $P'$ belong to $W$, there must be an internal vertex $a \in V(M_0)$ of both paths that intersects an edge of $W$ in both paths. However, since $\deg_W(a) \leq 1$, the edges of $W$ intersecting $a$ in both paths must be identical, which implies that the paths $P'$ and $P_1$ intersect also on some end-point. Since $P'$ and $P_1$ where chosen as general paths of $\cP'$ and $\cP_1$, respectively, that intersect on an internal node, this implies that the difference $|\cP'| - |\cP''|$ is equal to the number of paths in $\cP'$ that intersect a path of $\cP_1$ in an \emph{end-point}. The rest of the proof is devoted to proving that the last number is at most $6|\cP_1|$.

Since each path of $\cP_1$ has only two end points, to prove that the paths of $\cP_1$ intersect at most $6|\cP_1|$ paths of $\cP'$ at an end-point, it suffices to show that every vertex of $V \setminus V(M_0)$ can appear in at most $3$ paths of $\cP'$. To see why that is the case, consider an arbitrary vertex $u \in V \setminus V(M_0)$. If $u$ belongs to some path $P' \in \cP'$, then it must be in one of two roles as follows.
\begin{itemize}
	\item If $u$ is the last vertex of the path, then the last edge of the path is an edge $e \in W$ that includes $u$, and the other edges of the path $P'$ are the single edge of $M_0$ intersecting $e$ and the single edge of $W_M$ intersecting $e$. Note that this means that the identity of the entire path is determined by the edge $e$, and therefore, the number of paths of $\cP'$ in which $u$ is the last vertex can be upper bounded by $\deg_W(u) \leq 2$.
	\item If $u$ is the first vertex of the path, then the first edge of the path is the single edge $e \in W_M$ that includes $u$, and the other edges of the path are the single edge $e' \in M_0$ that intersect $e$ and the single edge $e'' \in W$ that intersects $e'$. Hence, the entire path is determined by the fact that $u$ is its first vertex, and therefore, there can be only a single path in $\cP'$ in which $u$ is the first vertex. \qedhere
\end{itemize}
\end{proof}

Originally, all the paths of $\cP''$ can be picked in the third pass of {\trianglefreealg} (Algorithm~\ref{alg:third_pass}) since they are vertex disjoint from the paths of $\cP_1$. However, as Algorithm~\ref{alg:third_pass} starts to add paths to $\cP_2$, it stops being possible to add some paths of $\cP''$ to $\cP_2$. Still, we can lower bound the size of $\cP_2$ in terms of the size of $\cP''$.
\begin{lemma} \label{lem:cP2_size}
$|\cP_2| \geq \tfrac{1}{6}|\cP''| \geq \tfrac{5}{9}|M^*| - \tfrac{8}{9}|M_0| - |\cP_1|$.
\end{lemma}
\begin{proof}
We begin the proof by observing that no edge $e \in M_0$ is connect by two distinct wings $w_1, w_2 \in W$ to vertices of $V \setminus (V(M_0) \cup V(\cP_1))$. Assume towards a contradiction that this is not true, then there is an edge $e$ in $G_A$ corresponds to the path $P$ defined as $w_1, e, w_2$. Since $M_A$ is a maximum matching in $G_A$, it must include at least one edge that contains some end-point of $P$ (otherwise, the edge corresponding to $P$ could be added to $M_A$, which violates its maximality); which contradicts the definition of either $w_1$ or $w_2$.

For every path $P'' \in \cP''$, let us charge a cost of $1$ to some path of $\cP_2$ that intersects it. To see why such a path must exist, let us denote by $e_M$ the edge of $P''$ that belongs to $W_M$ (the first edge of $P''$). When $e_M$ arrives, the path $P''$ was one candidate to be added to $\cP_2$ by Algorithm~\ref{alg:third_pass}. If this candidate was still feasible at this time (in the sense that it was vertex disjoint from $\cP_2$), then Algorithm~\ref{alg:third_pass} must have added either $P''$ to $\cP_2$ or another path that includes $e_M$. In either case, following the arrival of $e_M$, some path intersecting $P''$ (which is possibly $P''$ itself) appears in $\cP_2$---and can be charged.

Our next goal is to show that the total cost charged to any single path of $\cP_2$ is at most $6$, which implies the lemma because the total cost charged to all the paths of $\cP_2$ is exactly $|\cP''|$. We do that by making two observations.
\begin{itemize}
	\item Since $\cP'' \subseteq \cP'$, we get by the proof of Lemma~\ref{lem:cPDoublePrime_size} that at most $3$ paths of $\cP''$ can include any given vertex $u \in V \setminus V(M_0)$.
	\item Our second observation is that, if a path $P'' \in \cP''$ intersects a path $P_2 \in \cP_2$, then they must intersect on an end-point of $P_2$. Assume towards a contradictions that they only intersect on an internal node $a$. Since the middle edges of both paths are edges of $M_0$ that include $a$, both internal edges must be the same. Let us denote this internal edge by $e$. Furthermore, as explained above, there can be only a single edge $w \in W$ that intersects $e$ and does not include a vertex of $V(\cP_1)$. This edge must belong also to both paths, and therefore, the end-point of $w$ that does not belong to $V(M_0)$ is an end-point of both $P''$ and $P_2$.
\end{itemize}
Combining the above two observations, we get that, for every path $P_2 \in \cP_2$, only paths of $\cP''$ intersecting an end-point of $P_2$ can charge a cost to $P_2$, and there can be at most $3$ paths of $\cP''$ intersecting each such end-point. Since $P_2$ has only two end-points, this implies that at most $6$ paths of $\cP''$ can charge $P_2$.
\end{proof}

\begin{corollary} \label{cor:approximation_three_passes}
The size of the output of {\trianglefreealg} is $|M_0| + |\cP_1| + |\cP_2| \geq \tfrac{11}{18}|M^*| = (\tfrac{1}{2} + \tfrac{1}{9})|M^*|$.
\end{corollary}
\begin{proof}
The size of the output of {\trianglefreealg} is $|M_0| + |\cP_1| + |\cP_2|$ by Lemma~\ref{lem:basic_matching_size}, thus, we only need to lower bound this sum. To do this, note that
\begin{align*}
	|M_0| + |\cP_1| + |\cP_2|
	\geq{} &
	|M_0| + |\cP_1| + \{\tfrac{5}{9}|M^*| - \tfrac{8}{9}|M_0| - |\cP_1|\}\\
	={} &
	\tfrac{5}{9}|M^*| + \tfrac{1}{9}|M_0|
	\geq
	\tfrac{5}{9}|M^*| + \tfrac{1}{18}|M^*|
	=
	\tfrac{11}{18}|M^*|
	\enspace,
\end{align*}
where the first inequality follows from Lemma~\ref{lem:cP2_size}, and the second inequality follows from the observation made at the beginning of this section (namely, that $|M_0|$ is a $\nicefrac{1}{2}$-approximation for $|M^*|$ because $M_0$ is a maximal matching).
\end{proof}

Theorem~\ref{thm:improvement_paths_triangle_free} now follows from Corollaries~\ref{cor:space_complexity} and~\ref{cor:approximation_three_passes}.

%% file: GeneralImprovementPaths.tex
\section{Three-Pass Algorithm for General Graphs} \label{sec:general}

In this section we prove Theorem~\ref{thm:improvement_paths_general}, which we repeat here for convenience.
\ThmImprovementPathsGeneral*

The algorithm that we use to prove Theorem~\ref{thm:improvement_paths_general} is given as Algorithm~\ref{alg:algo_general}. Since this algorithm is very similar to the algorithm {\trianglefreealg} presented in Section~\ref{sec:triangle_free}, we use below the terminology and notation defined in the last section.

Intuitively, the reason why {\trianglefreealg} does not apply to general graphs is that given an edge $(a, b) \in M_0$, a wing $(u, a) \in W_M$ and a wing $(b, v) \in W$, we are not guaranteed that these three edges form an augmenting path for the matching $M_0$ because they might represent a triangle. To overcome this hurdle, Algorithm~\ref{alg:algo_general} constructs two sets of edges in its second pass: a set $W_1$ constructed exactly like the set $W$ in {\trianglefreealg}, and a set $W_2$ constructed in the same way, but while excluding the edges of $W_1$. Since $W_1$ and $W_2$ are disjoint, given an edge $(a, b) \in M_0$ and a wing $(u, a) \in W_M$, at most one of the sets $W_1$ or $W_2$ can contain a wing that forms a triangle together with these two edges, which intuitively allows us to bound the deterioration in the approximation guarantee resulting from the existence of such triangles.

\begin{algorithm}[th]
\caption{\textsc{Maximum Matching via Augmenting Paths -- General Graphs}} \label{alg:algo_general}
\tcp{First Pass}
Let $M_0 \gets \varnothing$.\\
\For{every edge $e$ that arrives}
{
	Add $e$ to $M_0$ if it does not intersect any edge that already belongs to $M_0$.
}
\BlankLine
\tcp{Second Pass}

Let $W_1 \gets \varnothing$, $W_2 \gets \varnothing$.\\
\For{every edge $e$ that arrives}
{
	\If{$e$ intersects exactly one vertex $u \in V(M_0)$}
	{
		Let $v$ denote the other end-point of $e$ (\ie, the end-point that is not $u$).\\
		\If{$\deg_{W_1}(u) < 1$ and $\deg_{W_1}(v) < 2$}
		{
			Add $e$ to $W_1$.
		}
		\ElseIf{$\deg_{W_2}(u) < 1$ and $\deg_{W_2}(v) < 2$}
		{
			Add $e$ to $W_2$.
		}

	}
}

\BlankLine

\tcp{Post-processing}
Let $G_A$ be a multi-graph over the vertices $V \setminus V(M_0)$. For every path $P_{u, v}$ of length $3$ in $W_1 \cup W_2 \cup M_0$ between two distinct vertices $u, v \in V \setminus V(M_0)$, we add an edge $(u, v)$ to the graph $G_A$. \tcp{This is a multi-graph because there might be multiple such paths between a pair of vertices of $V \setminus V(M_0)$.}
Find a maximum size matching $M_A$ in $G_A$.\\
Let $\cP_1 \gets \{P_{u, v} \mid (u, v) \in M_A\}$.

\BlankLine
\tcp{Third Pass}

Let $\cP_2 \gets \varnothing$.\\
\For{every edge $w'$ that arrives}
{
	\If{there exist $4$ vertices $u, a, b, v \in V \setminus (V(\cP_1) \cup V(\cP_2))$ such that: (i) $u \not \in V(M_0)$, (ii) $w' = (u, a)$, (iii) $(a, b) \in M_0$, (iv) $(b, v) \in W_1 \cup W_2 $ and (v) $u \neq v$}
	{
		Add the path $u, a, b, v$ to $\cP_2$. 
	}
}

\BlankLine

\Return{$M_0 \oplus \left(\bigcup_{P \in \cP_1 \cup \cP_2} E(P)\right)$}.
\end{algorithm}

We note that the analysis of {\trianglefreealg} up to Lemma~\ref{lem:W_lower_bound} applies to Algorithm~\ref{alg:algo_general} with two differences.
\begin{itemize}
	\item The proof of Corollary~\ref{cor:space_complexity} upper bounds by $n$ the size of the set $W$ (recall that this set is identical to the set $W_1$ in Algorithm~\ref{alg:algo_general}). To make this proof apply to Algorithm~\ref{alg:algo_general} as well, we need to observe that the size of the set $W_2$ is at most $n$ due to the same argument. In particular, this implies that Algorithm~\ref{alg:algo_general} is a semi-streaming algorithm.
	\item Lemma~\ref{lem:W_lower_bound} provides a lower bound on the size of the set $W$, which translates into an identical lower bound on the size of the corresponding set $W_1$ in Algorithm~\ref{alg:algo_general}.
\end{itemize}
 
In the rest of this section, it will be convenient to work with the set $W'_2$ constructed by Algorithm~\ref{alg:W'_2} (note that Algorithm~\ref{alg:W'_2} is used for analysis purposes only). Intuitively, $W'_2$ is constructed in the same general way in which $W_1$ and $W_2$ are constructed; however, while all the edges of the input stream are considered in the construction of $W_1$, and only the edges of $E \setminus W_1$ are considered in the construction of $W_2$, the construction of $W'_2$ takes into account the edges of $(E \setminus W_1) \cup W_M$.

\begin{algorithm}
\caption{Construction of $W'_2$}\label{alg:W'_2}
Let $W'_2 \gets W_2$.\\
\For{every edge $(u, v) \in W_1 \cap W_M$}
{
	Assume without loss of generality that $u$ is the end point of $(u, v)$ that belongs to $V(M_0)$.\\
	\If{$\deg_{W'_2} (u) < 1$ and $\deg_{W'_2} (v) < 2$}
	{
		Add $(u,v)$ to $W'_2$.
	}
}
\end{algorithm}

Since $W'_2$ is a subset of $W_1 \cup W_2$ by construction, the set $W_1 \cup W_2$ that is often referred to by Algorithm~\ref{alg:algo_general} is identical to the set $W_1 \cup W'_2$. Furthermore, one can observe that the lower bound proved by Lemma~\ref{lem:W_lower_bound} for $W_1$ applies also to $W'_2$ because all the edges of $W_M$ are considered for addition to $W'_2$ at some point (either during the construction of $W_2$ or in Algorithm~\ref{alg:W'_2}). This implies the following observation.

\begin{observation} \label{obs:two_W_size}
$|W_1| + |W'_2| \geq \tfrac{4}{3}|W_M|$.
\end{observation}

We now define a multi-set $\cP'$ similar to the set of the same name used in the analysis of {\trianglefreealg}. Specifically, $\cP'$ includes every triangle or path obtained by combining an edge $(u, a) \in W_M$, an edge $(a, b) \in M_0$ and an edge $(b, v)$ of either $W_1$ or $W'_2$. Moreover, if there are multiple options to obtain a path or triangle in this way, then the multiplicity of the path or triangle in $\cP'$ will be equal to the number of these options. To make this point clearer, we provide a pseudocode for constructing $\cP'$ as Algorithm~\ref{alg:P_prime} (again, Algorithm~\ref{alg:P_prime} is used for analysis purposes only).

\begin{algorithm}
\caption{Construction of $\cP'$}\label{alg:P_prime}
Let $\cP' \gets \varnothing$.\\
\For{every edge $(u, a) \in W_M$}
{
	\For{every edge $(a, b) \in M_0$}
	{
		\For{every edge $(b, v) \in W_1$}
		{
			Add the path/triangle $(u, a), (a, b), (b, v)$ to $\cP'$. \label{line:W_1_additions}
		}
		\For{every edge $(b, v) \in W'_2$}
		{
			Add the path/triangle $(u, a), (a, b), (b, v)$ to $\cP'$. \label{line:W'_2_additions}
		}
	}
}
\end{algorithm}

\begin{observation} \label{obs:cPprime_size_general}
$|\cP'| \geq \tfrac{20}{3}|M^*| - \tfrac{32}{3}|M_0|$.
\end{observation}
\begin{proof}
Repeating the proof of Observation~\ref{obs:cPprime_size}, we get that at least $|W_1| + |W_M| - |V(M_0)|$ paths are added to $\cP'$ in Line~\ref{line:W_1_additions} of Algorithm~\ref{alg:P_prime}, and at least $|W'_2| + |W_M| - |V(M_0)|$ paths are added to $\cP'$ in Line~\ref{line:W'_2_additions} of Algorithm~\ref{alg:P_prime}. Therefore,
\begin{align*}
	|\cP'|
	\geq{} &
	|W_1| + |W'_2| + 2|W_M| - 2|V(M_0)|
	\geq
	(\tfrac{4}{3} + 2)|W_M| - 2|V(M_0)|\\
	\geq{} &
	2(\tfrac{4}{3} + 2)(|M^*| - |M_0|) - 2|V(M_0)|
	=
	\tfrac{20}{3}|M^*| - \tfrac{32}{3}|M_0|
	\enspace,
\end{align*}
where the second inequality follows from Observation~\ref{obs:two_W_size}, and the last inequality follows from Observation~\ref{obs:W_M_size}.
\end{proof}

An element (path or triangle) of $\cP'$ has a potential to be added to $\cP_2$ by Algorithm~\ref{alg:algo_general} only if it is a path (\ie, not a triangle) and none of its vertices appears in $\cP_1$. Let $\cP''$ be the multi-set of such paths. The following lemma lower bounds the size of $\cP''$.
\begin{lemma}\label{lem:cP''_size}
$|\cP''| \geq |\cP'| - 12|\cP_1| - |M_0| \geq \tfrac{20}{3}|M^*| - \tfrac{35}{3}|M_0| - 12|\cP_1|$.
\end{lemma}
\begin{proof}
The second inequality of the lemma follows from Observation~\ref{obs:cPprime_size_general}, and therefore, we concentrate on proving the first inequality. Let $\tilde{\cP}'$ be the multi-set of paths/triangles from $\cP'$ that do not intersect any vertex of $\cP_1$. Repeating the proof of Lemma~\ref{lem:cPDoublePrime_size}, we get that $\tilde{\cP}'$ contains all the paths/triangles added to $\cP'$ by Line~\ref{line:W_1_additions} of Algorithm~\ref{alg:P_prime} except for up to $6|\cP_1|$ paths/triangles, and the same is true for the paths/triangles added to $\cP'$ by Line~\ref{line:W'_2_additions} of Algorithm~\ref{alg:P_prime}. Since every path/triangle in $\cP'$ was added to this mutli-set by either Line~\ref{line:W_1_additions} or Line~\ref{line:W'_2_additions} of Algorithm~\ref{alg:P_prime}, we get
\[
	|\tilde{\cP}'|
	\geq
	|\cP'| - 12|\cP_1|
	\enspace.
\]

Since $\cP''$ includes every path of $\tilde{\cP}'$, to complete the proof of the lemma it remains to show that $\tilde{\cP}'$ contains at most $|M_0|$ triangles. To see that this is indeed the case, we recall that every triangle (or path) in $\tilde{\cP}'$ must include a single edge of $M_0$, and we claim that no two triangles in $\tilde{\cP}'$ can share this edge (and therefore, the number of triangles is upper bounded by the number of edges in $M_0$). Assume towards a contradiction that this claim does not hold, \ie, that there exist two triangles $T_1,T_2 \in \tilde{\cP}'$ sharing an edge $e \in M_0$. Each one of these triangles must include one edge of $W_M$. Let $e_1$ and $e_2$ denote the edges of $W_M$ in $T_1$ and $T_2$, respectively, and let $e'_1$ the single edge of $T_1$ which is not $e$ or $e_1$ and $e'_2$ be the single edge of $T_2$ which is not either $e$ or $e_2$. We now need to consider two cases. The first case is when $e_1 = e_2$. In this case $e'_1$ and $e'_2$ must be also identical, and cannot belong to $W_M$ because $e_1 = e_2$ belongs to $W_M$ and $W_M$ is a subset of the matching $M^*$. However, this leads to a contradiction because one of the edges $e'_1$ or $e'_2$ must belong to $W_1$, and the other of these edges must belong to $W'_2$, and the sets $W_1$ and $W'_2$ can intersect only on edges of $W_M$.

It remains to consider the case in which $e_1 \neq e_2$. Let $u_1, u_2$ be the end-points of these edges, respectively, that do not belong to the edge $e$ of $M_0$. Since $e_1 \neq e_2$ are edges of the $W_M$, which is a subset of the matching $M^*$, $u_1$ and $u_2$ must be distinct. Consider now the path $e'_1, e, e'_2$. One can observe that this is indeed a path because (i) $u_1 \neq u_2$ and (ii) the fact that $e_1$ and $e_2$ are vertex disjoint implies that $e'_1$ and $e'_2$ intersect different end-points of $e$. Furthermore, since $T_1, T_2 \in \tilde{\cP}'$, this path does not intersect any vertex of $\cP_1$, and thus, its existence contradicts the maximality of the matching $M_A$ constructed by Algorithm~\ref{alg:algo_general} because both $e'_1$ and $e'_2$ belong to $W_1 \cup W'_2 = W_1 \cup W_2$.
\end{proof}

We are now ready to lower bound the number of augmenting paths found by Algorithm~\ref{alg:algo_general} during its third pass.
\begin{lemma}\label{lem:P_2_size_general}
$|P_2| \geq |\cP''|/12 \geq \tfrac{5}{9}|M^*| - \tfrac{35}{36}|M_0| - |\cP_1|$.
\end{lemma}
\begin{proof}
The proof of the lemma is very similar to the proof of Lemma~\ref{lem:cP2_size}, except that now every path of $\cP_2$ might get a charge of up to $12$ because the paths of $\cP''$ originally added to $\cP'$ by Line~\ref{line:W_1_additions} of Algorithm~\ref{alg:P_prime} can contribute up to $6$ to this charge, and the same goes for the paths of $\cP''$ originally added to $\cP'$ by Line~\ref{line:W'_2_additions} of this algorithm.
\end{proof}

Theorem~\ref{thm:improvement_paths_general} now follows from Corollary~\ref{cor:space_complexity} and the next corollary.
\begin{corollary}
The size of the matching produced by Algorithm~\ref{alg:algo_general} is at least $(\tfrac{1}{2} + \tfrac{1}{14.4})|M^*|$.
\end{corollary}
\begin{proof}
By Lemma~\ref{lem:basic_matching_size}, the size of the matching produced by Algorithm~\ref{alg:algo_general} is at least
\[
	|M_0| + |\cP_1| + |\cP_2|
	\geq
	\tfrac{5}{9}|M^*| + \tfrac{1}{36}|M_0|
	\geq
	\tfrac{5}{9}|M^*| + \tfrac{1}{72}|M^*|
	=
	(\tfrac{1}{2} + \tfrac{1}{14.4})|M^*|
	\enspace,
\]
where the first inequality holds by Lemma~\ref{lem:P_2_size_general}, and the second inequality holds since $M_0$ (as a maximal matching) is of size at least $\frac{1}{2}|M^*|$.
\end{proof}

%% file: ThreePassComponents.tex
\section{Three-Pass Non-{\MMF} Algorithm} \label{app:three_pass}

In this section we prove Theorem~\ref{thm:three_pass}, which we repeat below for convenience.
\ThmThreePass*

The algorithm used for proving Theorem~\ref{thm:three_pass} is a modified version of Algorithm~\ref{alg:triangles} that appears as Algorithm~\ref{alg:3_passes} and manages to obtain an improved approximation ratio at the cost of making an additional pass (\ie, it makes $3$ passes). The first pass of Algorithm~\ref{alg:3_passes} is identical to the first pass of Algorithm~\ref{alg:triangles}, however, the second and third passes of Algorithm~\ref{alg:3_passes} each consider only one of the two kinds of edges considered together in the second pass of Algorithm~\ref{alg:triangles}. To describe this in more details we use the terminology defined in Section~\ref{sec:two_passes} for describing Algorithm~\ref{alg:triangles}. In the second pass of Algorithm~\ref{alg:3_passes}, we construct a set $A_1$ in the same way in which this is done by Algorithm~\ref{alg:triangles}, \ie, by greedily adding to $A_1$ edges that connect a connection vertex of a {\naive} partial triangle with an isolated vertex. Then, in the third pass of Algorithm~\ref{alg:3_passes}, we greedily collect into another set, termed $A_2$, edges that connect connection vertices of two distinct {\naive} partial triangles. We stress that the construction of $A_2$ by Algorithm~\ref{alg:3_passes} is slightly different compared to the construction of the set carrying the same name in Algorithm~\ref{alg:triangles}. 
Upon termination of its third pass, Algorithm~\ref{alg:3_passes} outputs a maximum matching in the set of all the edges that it kept.

\begin{algorithm}[th]
\caption{\textsc{Maximum Matching via Greedy Triangles - 3 passes}} \label{alg:3_passes}
\tcp{First Pass}
Let $P \gets \varnothing$.\\
\For{every edge $e$ that arrives}
{
	\If{every connected component of the graph $(V, P \cup \{e\})$ is either a path of length at most $2$ or a triangle (cycle of size $3$)}
	{
		Add $e$ to $P$.
	}
}

\BlankLine

\tcp{Second Pass}
Let $A_1 \gets \varnothing$.\\
\For{every edge $(u, v) \not \in P$ that arrives}
{
	Let $C_u$ and $C_v$ be the connected components of $u$ and $v$, respectively, in $(V, P)$. We assume without loss of generality that $|C_u| > 1$, otherwise we swap the roles of $u$ and $v$. \tcp{Note that we cannot have $|C_u| = |C_v| = 1$ because the edge $(u, v)$ was not added to $P$ in the first pass.}
	\If{no edge of $A_1$ intersects $C_u$ and $C_v$, $|C_v| = 1$ and $u$ is a connection vertex of $C_u$}
	{
		Add the edge $(u, v)$ to $A_1$.
	}
}

\BlankLine

\tcp{Third Pass}
Let $A_2 \gets \varnothing$.\\
\For{every edge $(u, v) \not \in P \cup A_1$ that arrives}
{
	Let $C_u$ and $C_v$ be the connected components of $u$ and $v$, respectively, in $(V, P)$. We assume without loss of generality that $|C_u| > 1$, otherwise we swap the roles of $u$ and $v$. \tcp{Again, we cannot have $|C_u| = |C_v| = 1$.}
	\If{no edge of $A_1 \cup A_2$ intersects $C_u$ and $C_v$, and $u$ and $v$ are connection vertices of $C_u$ and $C_v$, respectively}
	{
		Add the edge $(u, v)$ to $A_2$.
	}
}

\BlankLine

\Return{a maximum matching in the graph $(V, P \cup A_1 \cup A_2)$}.
\end{algorithm}

The proof of Observation~\ref{obs:streaming} applies to Algorithm~\ref{alg:3_passes} as well, and therefore, Algorithm~\ref{alg:3_passes} is a semi-streaming algorithm. Below we concentrate on analyzing the approximation guarantee of this algorithm. It is important to note that the analysis of the approximation ratio of Algorithm~\ref{alg:triangles} up to Lemma~\ref{lem:dn_lower_bound} only depends on the behavior of the algorithm during its first pass, and therefore, applies also to Algorithm~\ref{alg:3_passes} since the two algorithms have identical first passes.

In principle, the proof of Lemma~\ref{lem:dn_greedy} applies also to Algorithm~\ref{alg:3_passes} since this proof is based on the method used by Algorithm~\ref{alg:triangles} to construct the set $A_1$, and this set is constructed in the same way by the two algorithms. However, it turns out that we need in this section a slightly stronger version of Lemma~\ref{lem:dn_greedy}. Specifically, Lemma~\ref{lem:dn_greedy} includes the value $\ftriangles$ in one of its terms. This value counts the number of connected components in $(V, P)$ that are triangles and do not include within them any edge of $M^*$. Each such connected component is intersected by at most a single edge of $A_1$ or $A_2$, and in this section we need to count separately the connected components of this kind that intersect edges from each one of these sets. Formally, we let $\ftrianglesdn$ be the number of connected components of $(V, P)$ that (1) are triangles, (2) do not include any edge of $M^*$, and (3) intersect an edge of $A_1$. Similarly, $\ftrianglesdd$ is the number of connected components of $(V, P)$ that (1) are triangles, (2) do not include any edge of $M^*$, and (3) intersect an edge of $A_2$. Since every partial triangle in $(V, P)$ intersects at most a single edge of $A_1 \cup A_2$, and the sets $A_1$ and $A_2$ are disjoint, we immediately get from these definitions $\ftrianglesdn + \ftrianglesdd \leq \ftriangles$. Furthermore, it is not difficult to verify that the proof of Lemma~\ref{lem:dn_greedy} in fact implies the following stronger version of the lemma.

\begin{lemma}[Stronger version of Lemma~\ref{lem:dn_greedy}] \label{lem:dn_greedy_strong}
\[3|A_1| \geq \dnpotential - \ftrianglesdn\enspace.\]
\end{lemma}

Lemma~\ref{lem:dn_greedy_strong} lower bounds the size of the set $A_1$. Our next objective is to find a lower bound also for the size of $A_2$. As a first step towards this goal, we upper bound the number of edges that have a potential to be added to $A_2$ immediately after the first pass of Algorithm~\ref{alg:3_passes}, but are removed from this potential during the second pass of the algorithm. To formalize this notion, let us recall that $\ddpotential$ is the set of edges of $M^*$ that connect connection vertices of two distinct partial triangles of $(V, P)$. Intuitively, $\ddpotential$ counts edges that have a potential to be added to $A_2$; however, for such an edge to really end up in $A_2$, it is required that the two partial triangles it intersect remain {\naive} after the second pass. Therefore, the size of the ``lost potential'' is the number of edges that are counted by $\ddpotential$, but intersect at least one partial triangle of $(V, P)$ that is also intersected by an edge of $A_1$. In the following, we denote this number by $\misseddd$.

\begin{lemma} \label{lem:size_of_misseddd}
\[
	\misseddd
	\leq
	3|A_1| - \dnpotential + \ftrianglesdn
	\enspace.
\]
\end{lemma}
\begin{proof}
The proof of this lemma is similar to the proof of Lemma~\ref{lem:dn_greedy}, however, we write it fully for completeness.

We say that an edge $e$ of $M^*$ counted by $\dnpotential$ or $\ddpotential$ is excluded by an edge $f \in A_1$ if $e$ and $f$ intersect the same connected component of $(V, P)$. One can observe that every edge $e$ counted by $\dnpotential$ is excluded by some edge of $A_1$ (possibly itself) when Algorithm~\ref{alg:3_passes} terminates because otherwise Algorithm~\ref{alg:3_passes} would have added $e$ to $A_1$, which would have resulted in $e$ excluding itself. Therefore, the number of edges counted by $\ddpotential$ that are excluded by some edge of $A_1$, which is exactly $\misseddd$, can be upper bound by the difference $|J| - \dnpotential$, where $J$ is the set of edges counted by either $\dnpotential$ or $\ddpotential$ that are excluded by the edges of $A_1$. In other words,
\begin{equation} \label{eq:exclusion_leftover}
	\misseddd
	\leq
	|J| - \dnpotential
	\enspace.
\end{equation}

Let $(u, v)$ be an edge of $A_1$, and assume without loss of generality that $v$ is the end point of this edge which is an isolated vertex of $(V, P)$. This implies that $u$ is a connection vertex of a connected component $C_u$ of $(V, P)$ which is either a path of length $2$ or a triangle. If $C_u$ is a path of length $2$, then the edge $(u, v)$ can exclude only edges counted by either $\dnpotential$ or $\ddpotential$ that intersect either $v$ or a connection vertex of $C_u$, and there can be only $3$ such edges because $M^*$ is a matching. Next, consider the case in which $C_u$ is a triangle which is not counted by $\ftriangles$. In this case there can be at most $2$ edges of $M^*$ intersecting $C_u$, and therefore, even though $(u, v)$ can exclude any edge of $\dnpotential$ or $\ddpotential$ intersecting $C_u$ or $v$, there can be only $3$ such edges. It remains to consider the case in which $C_u$ is a triangle counted by $\ftriangles$. In this case, $(u, v)$ can again exclude every edge of $\dnpotential$ or $\ddpotential$ that intersects $C_u$ or $v$, and this time there can be at most $4$ such edges. Combining all the above, we get that the number $|J|$ of edges excluded by all the edges of $A_1$ is at most
\begin{align*}
	&
	3|A_1| + |\{e \in A_1 \mid \text{$e$ intersects a triangle counted by $\ftriangles$}\}|\\
	={} &
	3|A_1| + \ftrianglesdn
	\enspace,
\end{align*}
where the equality holds because a triangle counted by $\ftriangles$ is counted also by $\ftrianglesdn$ if and only if some edge of $A_1$ intersects it. Plugging the last upper bound on $|J|$ into Inequality~\eqref{eq:exclusion_leftover} completes the proof of the lemma.
\end{proof}

We can now prove the promised lower bound on the size of $A_2$. 

\begin{lemma} \label{lem:size_of_A2}
\[
	4|A_2| \geq \ddpotential - \misseddd - \ftrianglesdd
	\enspace.
\]
\end{lemma}
\begin{proof}
Recall that $\misseddd$ counts a subset of the edges that are counted by $\ddpotential$. Let $D$ be the set of edges (of $M^*$) counted by $\ddpotential$ but not by $\misseddd$. We say that an edge $e \in D$ is excluded by an edge $f \in A_2$ if $e$ and $f$ intersect the same connected component of $(V, P)$. One can observe that every edge $e \in D$ is excluded by some edge of $A_2$ (possibly itself) when Algorithm~\ref{alg:3_passes} terminates because otherwise Algorithm~\ref{alg:3_passes} would have added $e$ to $A_2$, which would have resulted in $e$ excluding itself. Therefore, we can upper bound the size of $D$ by counting the number of edges excluded by the edges of $A_2$.

Let $(u, v)$ be an edge of $A_2$, and let $C_u$ and $C_v$ be the connected components of $(V, P)$ that include $u$ and $v$ respectively. Notice that since $(u, v) \in A_2$, both $C_u$ and $C_v$ must be either paths of length $2$ or triangles. The edge $(u, v)$ excludes every edge of $D$ that intersects either $C_u$ or $C_v$. The number of $D \subseteq M^*$ edges that intersect $C_u$ can be at most $2$, unless $C_u$ is a triangle counted by $\ftriangles$, in which case there might be $3$ edges of $D$ intersecting $C_u$. Since a similar claim applies to $C_v$, we get that the number of edges excluded by all the edges of $A_2$ is at most
\[
	4|A_2| + \sum_{e \in A_2} T(e)
	=
	4|A_2| + \ftrianglesdd
	\enspace,
\]
where $T(e)$ is the number of triangles counted by $\ftriangles$ that intersect $e$, and the equality holds since a triangle is counted by $\ftrianglesdd$ if and only if it is both counted by $\ftriangles$ and intersects an edge of $A_2$. As explained above, the last expression is an upper bound on the size of $D$. Therefore, we get
\begin{align*}
	\ddpotential - \misseddd
	={} &
	|D|\\
	\leq{} &
	4|A_2| + \ftrianglesdd
	\enspace.
\end{align*}
The lemma now follows by rearranging this inequality.
\end{proof}

\begin{corollary} \label{cor:A1_A2}
\begin{align*}
	12|A_1| + 12|A_2| \geq{} & 4\dnpotential + 3\ddpotential \\&- 4\ftrianglesdn - 3\ftrianglesdd
	\enspace.
\end{align*}
\end{corollary}
\begin{proof}
Plugging Lemma~\ref{lem:size_of_misseddd} into Lemma~\ref{lem:size_of_A2}, we get
\begin{align*}
	4|A_2|
	\geq{} &
	\ddpotential - \misseddd - \ftrianglesdd\\
	\geq{} &
	\ddpotential - (3|A_1| - \dnpotential + \ftrianglesdn) \\&- \ftrianglesdd
	\enspace.
\end{align*}
Rearranging the last inequality, and multiplying it by $3$, yields
\begin{align*}
	9|A_1| + 12|A_2|
	\geq{} &
	3\ddpotential + 3\dnpotential \\&- 3\ftrianglesdn - 3\ftrianglesdd
	\enspace.
\end{align*}
The corollary now follows by adding Lemma~\ref{lem:dn_greedy_strong} to the last inequality.
\end{proof}

Let us now define $L_2 = \scomp + \dcomp + \tcomp + |A_1| + |A_2|$. The following lemma shows that one can obtain an approximation guarantee for Algorithm~\ref{alg:3_passes} by lower bounding $L_2$. Since the proof of this lemma is very similar to the proof of Lemma~\ref{lem:L_is_output}, we omit it.
\begin{lemma} \label{lem:L_2_bounds_solution}
Algorithm~\ref{alg:3_passes} outputs a matching of size at least $L_2$.
\end{lemma}

It remains now to lower bound $L_2$, which we do in the next lemma. Together with Lemma~\ref{lem:L_2_bounds_solution} and the above observation that Algorithm~\ref{alg:3_passes} is a semi-streaming algorithm, this lemma completes the proof of Theorem~\ref{thm:three_pass}.

\begin{lemma} \label{lem:L2_lower_bound}
$L_2 \geq \nicefrac{5}{9} |M^*|$.
\end{lemma}
\begin{proof}
Observe that
\begin{align*}
	12L_2
	={} &
	12\scomp + 12\dcomp + 12\tcomp + 12|A_1| + 12|A_2|\\
	\geq{} &
	12\scomp + 12\dcomp + 12\tcomp + 4\dnpotential \\&+ 3\ddpotential - 4\ftrianglesdn \\&- 3\ftrianglesdd\\
	\geq{} &
	12\scomp + 12\dcomp + 12\tcomp + 4\dnpotential \\&+ 3\ddpotential - 4\ftriangles
	\enspace,
\end{align*}
where the first inequality follows from Corollary~\ref{cor:A1_A2}; and the second inequality holds since  we already observed that $\ftriangles \geq \ftrianglesdn + \ftrianglesdd$, and the value $\ftrianglesdd$ is non-negative by definition.

To further develop the last inequality, we recall that the analysis from Section~\ref{sec:two_passes} up until, and including, Lemma~\ref{lem:auxiliary_inequalities} applies to Algorithm~\ref{alg:3_passes} as well. Therefore,
\begin{align*}
	12L_2
	\geq{} &
	12\scomp + 12\dcomp + 12\tcomp + 4\dnpotential \\&+ 3\ddpotential - 4\ftriangles\\
	\geq{} &
	\tfrac{28}{3}\scomp + 4\dcomp + 4\tcomp + \tfrac{20}{3}\dnpotential \\&+ \tfrac{25}{3}\ddpotential + \tfrac{8}{3}\sspotential + 4\dspotential \\&+ \tfrac{8}{3}\dmpotential - 4\ftriangles\\
	\geq{} &
	\tfrac{28}{3}\scomp + \tfrac{20}{3}\dnpotential + \tfrac{25}{3}\ddpotential \\&+ \tfrac{8}{3}\sspotential + 4\dspotential + \tfrac{20}{3}\dmpotential\\
	\geq{} &
	\tfrac{20}{3}\dnpotential + \tfrac{25}{3}\ddpotential + 12\sspotential \\&+ \tfrac{26}{3}\dspotential + \tfrac{20}{3}\dmpotential\\
	\geq{} &
	\tfrac{20}{3}|M^*| + \tfrac{5}{3}\ddpotential + \tfrac{16}{3}\sspotential + 2\dspotential
	\enspace,
\end{align*}
where the second Inequality holds by Inequality~\eqref{eq:basic_inequality}, the third inequality follows from Inequality~\eqref{eq:auxilary_1} (of Lemma~\ref{lem:auxiliary_inequalities}), the fourth inequality follows from Inequality~\eqref{eq:auxilary_2} (of Lemma~\ref{lem:auxiliary_inequalities}), and the last inequality holds by Inequality~\eqref{eq:partition} (of Lemma~\ref{lem:auxiliary_inequalities}).

The lemma now follows by rearranging the last inequality and observing that $\sspotential$, $\ddpotential$ and $\dspotential$ are all non-negative values by definition.
\end{proof}